\documentclass[runningheads]{llncs}

\usepackage[T1]{fontenc}
\usepackage{graphicx}

\usepackage{amsmath}
\usepackage{amssymb}
\usepackage{bbold}
\usepackage{graphicx}
\usepackage{color}
\usepackage[pdftex,bookmarks,hidelinks,breaklinks]{hyperref} 
\usepackage{listings}
\usepackage{mathabx}
\usepackage{mathpartir}
\usepackage{stmaryrd}
\usepackage{mathtools}
\usepackage{upquote}
\usepackage[dvipsnames]{xcolor}
\usepackage[all]{xy}
\usepackage[capitalise,nameinlink]{cleveref}
\usepackage{ebproof}
\usepackage{enumitem}
\usepackage{fontawesome5}

\hypersetup{
    colorlinks = true,
    allcolors = [rgb]{0 0.43 0.72},
}

\begin{document}

% !TEX root = temporal-resources.tex

% Any macro that is actually used should have a comment explaining what it is for.
% Please fight macro pollution and remove the macros that are not used.

\newcommand{\defeq}{\mathrel{\overset{\text{\tiny def}}{=}}} % Definitional equality

\newcommand{\pl}[1]{\textsc{#1}} % the name of a programming language

\newcommand{\lambdaTau}{\lambda_{[\tau]}} % the name of the calculus

\newcommand{\ctxunlock}{\text{\small\faUnlock}} % unlock modality in Fitch-style modal lambda calculi

\newcommand{\rulename}[1]{{\sc{#1}}} % names of derivation rules

\newcommand{\dotminus}{\mathbin{\scriptstyle\dot{\smash{\textstyle-}}}} % truncated minus on naturals

% BNF grammars
\newcommand{\bnfis}{\mathrel{\;{:}{:}{=}\ }}
\newcommand{\bnfor}{\mathrel{\;\big|\ \ }}

%%%%% Semantic concepts

%%% Modalities

\newcommand{\boxmodstd}{\raisebox{-0.75mm}{\scalebox{2}{$\square$}}} % standard S4-style box modality

%%% Signatures

\newcommand{\op}{\mathsf{op}} % a generic operation symbol
\newcommand{\opto}{\leadsto} % the wiggly arrow for operation signature
\newcommand{\tysigop}[3]{#1 \opto #2 \mathbin{!} #3} % the full signature of an operation

%%% Monad

\newcommand{\T}{T} % a generic monad
\newcommand{\TS}{S} % an alternative generic monad
\newcommand{\etaT}{\eta^{\T}} % unit
\newcommand{\muT}{\mu^{\T}} % multiplication
\newcommand{\strT}{\mathsf{str}^{\T}} % strength

\newcommand{\opT}{\op^{\T}} % algebraic operations
\newcommand{\delayT}{\mathsf{delay}^{\T}} % delay operation

%%% PRA structure

\newcommand{\etaPRA}{\eta^{\text{PRA}}}
\newcommand{\epsPRA}{\varepsilon^{\text{PRA}}}

%%% Categorical models

\newcommand{\expto}{\Rightarrow} % set exponentiation
\newcommand{\lam}[1]{\lambda #1 \,.\,} % lambda abstraction

\newcommand{\One}{\mathbb{1}} % terminal object
\newcommand{\Zero}{\mathbb{0}} % initial object

\newcommand{\Cat}{\mathbb{C}} % generic category
\newcommand{\Set}{\mathsf{Set}} % category of sets
\newcommand{\Pshf}{\Set^{(\mathbb{N},\le)}} % the concrete presheaf category

\newcommand{\id}{\mathsf{id}} % identity morphism

\newcommand{\pastmod}[1]{\langle #1 \rangle} % semantics of the modality on contexts
\newcommand{\futuremod}[1]{[ #1 ]} % semantics of the modality on types

\newcommand{\epsP}{\varepsilon^{\langle\rangle}} % derived counit-like map for the past modality
\newcommand{\etaF}{\eta^{[]}} % derived unit-like map for the future modality

\newcommand{\fst}{\mathsf{fst}} % first projection
\newcommand{\snd}{\mathsf{snd}} % second projection

\newcommand{\m}{\mathsf{m}} % witness that [ - ] is monoidal

\newcommand{\e}{\mathsf{e}} % derived morphism extracting time of environment

\newcommand{\h}{\chi} % semantic handling map

\newcommand{\enr}{\mathsf{enr}} % [-]-enrichment

\newcommand{\etaA}{\eta^{\dashv}} % unit of the adjunction 
\newcommand{\epsA}{\varepsilon^{\dashv}} % counit of the adjunction 

%%%%% Contexts

\newcommand{\ctxmod}[1]{\langle #1 \rangle} % modality on contexts

\newcommand{\ctxtime}[1]{\mathsf{time}\; #1} % time encapsulated in a context
\newcommand{\ctxminus}{\mathbin{-}} % minus operation on contexts

%%%%% Types

\newcommand{\duration}{\mathbin{!}} % the ! sign, with proper spacing

%% Value types
\newcommand{\tybase}{\mathsf{b}} % a base type
\newcommand{\tyunit}{\mathsf{unit}} % the unit ground type
\newcommand{\tyempty}{\mathsf{empty}} % the empty ground type
\newcommand{\typrod}[2]{#1 \times #2} % product type
\newcommand{\tysum}[2]{#1 + #2} % sum type
\newcommand{\tyfun}[2]{#1 \to #2} % function type
\newcommand{\tybox}[2]{[ #1 ]\, #2} % function type

%% Computation types

\newcommand{\tycomp}[2]{#1 \duration #2} % computation type

%%%%% Display of source code in math mode

\newcommand{\Ops}{\mathcal{O}} % the set of all operation names

\newcommand{\tm}[1]{\mathsf{#1}} % the source code font
\newcommand{\tmkw}[1]{\tm{\color{keywordColor}#1}} % source code keyword, colored

\newcommand{\tmconst}[1]{\tm{#1}} % constants
\newcommand{\tmunit}{()} % the element of the unit type
\newcommand{\tmpair}[2]{( #1 , #2 )} % ordered pair
\newcommand{\tminl}[2][]{\tmkw{inl}_{#1}\,#2} % left injection
\newcommand{\tminr}[2][]{\tmkw{inr}_{#1}\,#2} % right injection
\newcommand{\tmfun}[2]{{\mathop{\tmkw{fun}}}\; (#1) \mapsto #2} % function abstraction
\newcommand{\tmapp}[2]{#1\,#2} % application
\newcommand{\tmbox}[2][]{\tmkw{box}_{#1}\,#2} % boxing

\newcommand{\tmreturn}[2][]{\tmkw{return}_{#1}\, #2} % pure computation
\newcommand{\tmlet}[4][]{\tmkw{let}_{#1}\; #2 = #3 \;\tmkw{in}\; #4} % let-binding
\newcommand{\tmhandle}[4]{\tmkw{handle}\; #1 \;\tmkw{with}\; #2 \;\tmkw{to}\; #3\;\tmkw{in}\; #4} % handling
\newcommand{\tmunbox}[4][]{\tmkw{unbox}_{#1}\; #2 \;\tmkw{as}\; #3 \;\tmkw{in}\; #4} % unboxing

\newcommand{\tmop}[3]{\tm{#1}\;#2\; #3} % operation call
\newcommand{\tmgeneff}[2]{\tm{#1}\; #2} % generic operation call
\newcommand{\tmcont}[2]{(#1 \,.\, #2)} % a continuation
\newcommand{\tmopclause}[3]{(#1 \,.\, #2 \,.\, #3)} % operation clause of a handler

\newcommand{\tmdelay}[2]{\tmkw{delay}\;#1\; #2} % operation call

\newcommand{\tmmatch}[3][]{\tmkw{match}\;#2\;\tmkw{with}\;\{#3\}_{#1}} % match statement

%%% Typing rules

\newcommand{\types}{\vdash} % typing judgement
\newcommand{\of}{\mathinner{:}} % the colon in a typing judgement

\newcommand{\vj}[3]{#1 \vdash #2 : #3} % value typing
\newcommand{\cj}[4]{#1 \vdash #2 : #3 \mathrel{!} #4} % user computation typing

\newcommand{\subty}[2]{#1 \sqsubseteq #2} % subtyping relation (old macro)
\newcommand{\sub}{\sqsubseteq} % subtyping relation

\newcommand{\mkrule}[3]{\frac{#1}{#3}{\textsc{#2}}} % temporary, until we sort out proof package.

\newcommand{\ruleinfer}[3]{\inferrule*[Lab={\color{rulenameColor}#1}]{#2}{#3}}

%%% Denotational semantics

\newcommand{\sem}[1]{[\![#1]\!]} % semantic bracket
\newcommand{\gsem}[1]{[\![#1]\!]^g} % semantic bracket for ground types

\newcommand{\cond}[3]{\mathsf{if}~#1~\mathsf{then}~#2~\mathsf{else}~#3} % single line conditional

%% TITLE
\title{When Programs Have to Watch Paint Dry}

%% AUTHOR
\author{Danel Ahman}
\institute{
  Faculty of Mathematics and Physics, University of Ljubljana, Ljubljana, Slovenia \\
  \email{danel.ahman@fmf.uni-lj.si}}

\maketitle

%% ABSTRACT
\begin{abstract}
We explore type systems and programming abstractions for the safe usage of resources. 
In particular, we investigate how to use types to modularly specify and check \emph{when} 
programs are allowed to use their resources, e.g., when programming a robot arm on a 
production line, it is crucial that painted parts are given enough 
time to dry before assembly. We capture such \emph{temporal resources} using a 
time-graded variant of Fitch-style modal type systems, develop a corresponding modally typed, effectful 
core calculus, and equip it with a graded-monadic denotational  
semantics illustrated by a concrete presheaf model. Our calculus also includes  
graded algebraic effects and effect handlers. They are given a novel temporally aware treatment in
which operations' specifications include their execution times and their continuations 
know that an operation's worth of additional time has passed before they start executing, 
making it possible to safely access further temporal resources in them.

\end{abstract}

%%%%%%%%%%%%%%%%%%%%%%%%%%%%%%%%%%%%%%%%%%%%%%%%%%%%

%% CODE SNIPPETS TYPESETTING

\definecolor{codegreen}{rgb}{0,0.6,0}
\definecolor{codegray}{rgb}{0.5,0.5,0.5}
\definecolor{codepurple}{rgb}{0.58,0,0.82}
\definecolor{backcolour}{rgb}{0.95,0.95,0.92}

\colorlet{keywordColor}{NavyBlue} % the color of language keywords
\colorlet{rulenameColor}{Gray} % the color of rule names

%%%%%%%%%%%%%%%%%%%%%%%%%%%%%%%%%%%%%%%%%%%%%%%%%%%%

%% PAPER CONTENTS

% Customize the display of references to sections, subsections, subsubsections, theorems, and propositions.
\crefformat{section}{\S#2#1#3}
\Crefformat{section}{\S#2#1#3}

\crefformat{subsection}{\S#2#1#3}
\Crefformat{subsection}{\S#2#1#3}

\crefformat{subsubsection}{\S#2#1#3}
\Crefformat{subsubsection}{\S#2#1#3}

\crefformat{theorem}{Thm.~#2#1#3}
\Crefformat{theorem}{Thm.~#2#1#3}

\crefformat{proposition}{Prop.~#2#1#3}
\Crefformat{proposition}{Prop.~#2#1#3}

\crefformat{figure}{Fig.~#2#1#3}
\Crefformat{figure}{Fig.~#2#1#3}

% !TEX root = temporal-resources.tex

\section{Introduction}
\label{sect:introduction}

The correct usage of resources is at the heart of many programs, especially if they 
control safety-critical machinery. 
Such resources can take many different forms: ensuring that file 
handles are not arbitrarily duplicated or discarded (as captured by linear  
and uniqueness types)~\cite{Benton:LinearLambda,Girard:LinearLogic,Koopman:FPinClean}, 
or guaranteeing that communication happens according to 
protocols (as specified by session types)~\cite{Honda:SessionTypes,Wadler:GV}, or controlling 
how data is laid out in memory (as in Hoare and separation 
logics)~\cite{Ahman:RecallingAWitness,Jung:Iris,Nanevski:HTT,Reynolds:SeparationLogic}, or 
assuring that resources are correctly finalised~\cite{Ahman:Runners,Leijen:Finalisation}.

In contrast to the above approaches that predominantly focus on \emph{how} resources 
are used, we study how to modularly specify and  
verify \emph{when} programs can use their resources---we call such resources 
\emph{temporal}. For instance, consider the following code snippet controlling 
a robot arm on a (car) production line:
\[
\small
\begin{array}{l}
\tmlet{(\text{body'},\text{left-door'},\text{right-door'})}{\tmgeneff{paint}{(\text{body},\text{left-door},\text{right-door})}}{
  \\
  \tmgeneff{assemble}{(\text{body'},\text{left-door'},\text{right-door'})}}
\end{array}
\]

Here, the correct execution of the program (and thus operation of the robot arm 
it is controlling) relies on the car parts given enough time to dry between painting and 
assembly. Therefore, in its current form, the above code is correct only if a 
compiler (or a scheduler) inserts enough of a time delay at compile time 
(resp.~dynamically blocks program's execution for enough time) between the calls to  
$\tm{paint}$ and $\tm{assemble}$. However, in either case, one still faces the 
question of how to reason about the correctness of the compiled code (resp.~dynamic checks).

In this paper, we focus on developing a type system based means for reasoning about the 
temporal correctness of the code that the above-mentioned compiler might produce, 
or that a programmer might write directly when full control of the 
code is important. 
In particular, we had \emph{three desiderata} we set out to fulfil:
\begin{enumerate}[itemsep=5pt]
\item We did not want the delay between $\tm{paint}$ and $\tm{assemble}$ 
  to be limited to just \emph{blocking execution}, with the robot sitting idly 
  while watching paint dry.
  Instead, we wanted a flexible formalism that would allow
  the robot to spend that time \emph{doing other useful work}, 
  while ensuring that enough time passes.
\item We wanted the \emph{passage of time of program execution to be modelled within 
the type system}, rather than being left to some unspecified meta-level run-time.
\item We wanted the resulting language to give programmers the freedom to \emph{redefine the 
  behaviour of operations} such as $\tm{paint}$ and $\tm{assemble}$, say, via \emph{effect 
  handling}~\cite{Plotkin:HandlingEffects}, 
  while respecting the operations' temporal specifications.
\end{enumerate}

\subsubsection*{Paper Structure}

We achieve these goals by designing a \emph{mathematically
natural core programming language} for safe and correct programming with temporal 
resources: on the one hand, based on a time-graded, temporal variant of \emph{Fitch-style 
modal type systems}~\cite{Clouston:FitchStyle,Gratzer:PRA}, and on the other hand, 
on \emph{graded monads}~\cite{Katsumata:GradedMonads,Mellies:GradedMonads,Smirnov:GradedMonads}. 

We review modal types and discuss how we use them to capture 
temporal resources in \cref{sect:overview}. In \cref{sect:core-calculus}, we 
present $\lambdaTau$---our modally typed, effectful, equationally presented core calculus 
for safe programming with temporal resources. We justify the design of $\lambdaTau$ by giving 
it a mathematically natural sound denotational semantics in \cref{sect:semantics}, based on graded 
monads and adjunctions between strong monoidal functors, including a concrete 
presheaf example. In \cref{sect:delay-equations}, we briefly discuss a 
specialisation of $\lambdaTau$ with equations for time delays. We review 
related work and remark on future work in \cref{sect:related-future-work}, 
and conclude in \cref{sect:conclusion}.
This paper is also accompanied by an online appendix 
(\href{https://arxiv.org/abs/2210.07738}{https://arxiv.org/abs/2210.07738})
that presents further details of renamings and denotational semantics that we omit  
in \cref{sect:core-calculus} and \cref{sect:semantics}.

For supplementary rigour, we have formalised the main  
results of \cref{sect:core-calculus} and \cref{sect:semantics} also in Agda~\cite{agda:aw}, available at
\href{https://github.com/danelahman/temporal-resources/releases/tag/fossacs2023}{https://github.com/danelahman/temporal-resources/releases/tag/fossacs2023}.
Regrettably, it currently lacks (i) proofs of some auxiliary lemmas 
noted in \cref{prop:substitution-semantic} due to a bug in Agda where \textsc{with}-abstractions produce ill-typed  
terms,\footnote{Eta-contraction is not type-preserving: \href{https://github.com/agda/agda/issues/2732}{https://github.com/agda/agda/issues/2732}}
and (ii) two laws of the presheaf model because 
unfolding of definitions produces unmanageably large terms for Agda.

% !TEX root = temporal-resources.tex

\section{Modal Types for Temporal Resources}
\label{sect:overview}

We begin with an overview of (Fitch-style) modal type systems and how a 
time-graded variant of them naturally captures temporal aspects of resources.

\subsection{(Fitch-Style) Modal Types}
\label{sect:overview-modal-types}

A \emph{modal type system} extends the types of an underlying type system with
new \emph{modal type formers},\footnote{For brevity, we use the term \emph{modal type system} to interchangeably refer to
both modal type systems and natural deduction systems of (intuitionistic) modal logics.} e.g., $\boxmodstd X$, which states
that the type is to be considered and reasoned about in a \emph{different mode} compared to $X$, 
which can take many forms. For instance, in Kripke's possible worlds semantics, $\boxmodstd X$ 
means that values of type $X$ are \emph{available in all 
future worlds}~\cite{Kripke:PossibleWorlds}; 
in run-time code generation,  
the type $\boxmodstd X$ captures \emph{generators of $X$-typed code}~\cite{Wickline:Staging}; and
in asynchronous and distributed programming, the type $\boxmodstd X$ 
specifies  
\emph{mobile $X$-typed values}~\cite{Ahman:HOAeff,Murphy.PhDThesis,Radescek:MScThesis}.

Many different approaches to presenting modal type systems have been developed, with one of the main 
culprits being the difficulty of getting the \emph{introduction rule} for $\boxmodstd X$ correct. 
Namely, bearing in mind Kripke's possible worlds semantics, the introduction rule for $\boxmodstd X$ 
must allow one to use only those hypotheses that also hold in all future worlds, while at the same time 
ensuring that the system still enjoys expected structural properties.
Solutions to this problem have involved proving $\boxmodstd X$ in a context containing only 
$\boxmodstd$-types~\cite{Prawitz:NaturalDeduction} (with a failure of structural properties in the 
naive approaches), or building a form of explicit substitutions into the introduction rule for $\boxmodstd X$ 
to give the rule premise access to only $\boxmodstd$-types~\cite{Bierman:ModalLogic}, 
or incorporating the Kripke semantics in the type system by explicitly
indexing types with worlds~\cite{Simpson:PhDThesis}---see~\cite{Kavvos:ModalSurvey} for an in-depth survey.

In this paper, we build on \emph{Fitch-style modal type systems}~\cite{Borghuis:PhDThesis,Clouston:FitchStyle,Gratzer:PRA,Martini:CompIntModalProofs}, 
where the typing rules for $\boxmodstd X$ are given with respect to another 
modality, \ctxunlock, that acts on contexts, resulting in a particularly pleasant type-theoretic
presentation.

As an illustrative example, in a Fitch-style modal type system corresponding to the 
modal logic S4 (whose Kripke models require the order on worlds to be reflexive and transitive, 
thus also corresponding to natural properties of time), the typing rules for variables and the 
$\boxmodstd X$ type have the following form:\footnote{Depending on which exact modal
 logic one is trying to capture, the form of contexts used in the introduction/elimination rules
 can differ, see~\cite{Clouston:FitchStyle} for a detailed overview.}
\[
  \ruleinfer{Var}{
    \ctxunlock \not\in \Gamma'
  }{
    \Gamma, x \of X, \Gamma' \types x : X
  }
  \qquad
  \ruleinfer{Shut}{
    \Gamma, \ctxunlock \types t : X
  }{
    \Gamma \types \tmkw{shut}\;{t} : \boxmodstd{X}
  }
  \qquad
  \ruleinfer{Open}{
    \Gamma \types t : \boxmodstd{X}
  }{
    \Gamma,\Gamma' \types \tmkw{open}\;{t} : X
  }
\]

Intuitively, the \emph{context modality} $\ctxunlock$ creates a barrier in the premise of \rulename{Shut} 
so that only $\boxmodstd$-typed variables can be used from $\Gamma$ in 
$t$, achieving the above-mentioned 
correctness goal for the introduction rule of $\boxmodstd X$. Alternatively, in the context of Kripke's possible 
worlds semantics, one can also read the occurrences of the $\ctxunlock$ modality as advancing 
the underlying world---in \rulename{Shut}, $t$ in the premise is typed in some future world compared 
to where $\tmkw{shut}\;{t}$ is typed at.
This intuition will be useful to how we use a similar modality  
to capture the passage of time in $\lambdaTau$. The context weakening $\Gamma,\Gamma'$ 
in \rulename{Open} ensures the admissibility of structural rules, and in the possible worlds reading, 
it intuitively expresses that if $\boxmodstd X$ is available 
in some world, then $X$ will be available in all possible future worlds.

\subsection{Modal Types for Temporal Resources}
\label{sect:overview-temporal-resources}

Next, we give a high-level overview of how we use a time-graded variant of Fitch-style 
modal type systems to capture temporal properties of resources  in $\lambdaTau$.
For this, we use the production line code snippet from \cref{sect:introduction} as 
a working example.

\subsubsection{A Naive Approach}

Before turning to modal types, a naive solution to achieve the desired time delay would be for 
$\tm{paint}$ to return the required drying time and for the program to delay
execution for that time duration, e.g., as expressed in
\[
\small
\begin{array}{l}
\tmlet{(\tau_{\text{dry}}, \text{body'},\text{left-door'},\text{right-door'})}{\tmgeneff{paint}{(\text{body},\text{left-door},\text{right-door})}}{
  \\
  \tm{delay}\;{\tau_{\text{dry}}};
  \\
  \tmgeneff{assemble}{(\text{body'},\text{left-door'},\text{right-door'})}}
\end{array}
\]
It is not difficult to see that we could generalise this solution to allow performing other 
useful activities while waiting for $\tau_{\text{dry}}$ time to pass. So are we done and can we 
conclude the paper here? Well, no, because this solution puts all the burden for writing 
correct code on the shoulders of the programmer, with successful typechecking giving no
additional guarantees that $\tau_{\text{dry}}$ indeed will have passed.

\subsubsection{A Temporal Resource Type}

Instead, inspired by Fitch-style modal type systems and Kripke's possible worlds semantics of 
the $\boxmodstd$-modality, we propose a \emph{temporal resource type}, 
written $\tybox{\tau}{X}$, to specify that a value of type $X$ will become available for use in \emph{at most} 
$\tau$ time units, or to put it differently, the boxed value of type $X$ can be explicitly unboxed only 
when \emph{at least} $\tau$ time units have passed.
Concretely, $\tybox{\tau}{X}$ is presented by the following two typing rules:
\[
  \ruleinfer{Box}{
    \Gamma, \ctxmod \tau \types V : X
  }{
    \Gamma \types \tmbox[\tau]{V} : \tybox{\tau}{X}
  }
  \hspace{1ex}
  \ruleinfer{Unbox}{
    \tau \le \ctxtime \Gamma \\
    \Gamma \ctxminus \tau \types V : \tybox{\tau}{X} \\
    \Gamma, x \of X \types N : \tycomp{Y}{\tau'}
  }{
    \Gamma \types \tmunbox[\tau]{V}{x}{N} : \tycomp{Y}{\tau'}
  }
\]
Above, $\tau$s are natural numbers that count discrete time moments, and 
$\tycomp{Y}{\tau'}$ is a type of 
computations returning $Y$-typed values and executing in $\tau'$ time units.

Analogously to the context modality $\ctxunlock$ of Fitch-style modal type systems, we  
introduce a similar \emph{modality on contexts}, written $\ctxmod \tau$, to express that when 
typechecking a term of the form $\Gamma, \ctxmod \tau \types V : X$, we can safely
assume that \emph{at least} $\tau$ time will have passed before $V$ is accessed or executed, 
as in the premise of the \rulename{Box} rule. Accordingly, in \rulename{Unbox}, 
we require that at least $\tau$ time units have passed since the resource $V$
of type $\tybox{\tau}{X}$ was created or brought into scope, by typing $V$  
in the ``earlier'' context $\Gamma \ctxminus \tau$ 
(we define this operation in \cref{sect:type-system}).

Encapsulating temporal resources as a type gives us
flexible first-class access to them, and allows to pack them 
in data structures and pass them to functions.

\subsubsection{Modelling Passage of Time}

As we see in the \rulename{Unbox} rule, we can unbox a temporal resource only
when enough time has passed since its creation. This begs the question: How can the 
passage of time be modelled within the type system? For this, we propose a new notion of
\emph{temporally aware graded algebraic effects}, where each operation $\op$ is 
specified not only by its parameter and result types, but also by its 
prescribed execution time, and with $\op$'s continuation knowing that $\op$'s
worth of additional time has passed before it begins executing. We refer the reader 
to~\cite{Bauer:WhatIsAlgebraic,Hyland:CombiningEffects,Katsumata:GradedMonads,Plotkin:NotionsOfComputation}
for background on ordinary (graded) algebraic effects.

For instance, the $\tm{paint}$ operation, taking $\tau_{\text{paint}}$ time, is typed in $\lambdaTau$ as\footnote{We
  present $\lambdaTau$ formally using algebraic operations with explicit continuations, 
  while in code snippets we use so-called \emph{generic effects}~\cite{Plotkin:AlgOperations} without explicit continuations.}
\[
  \ruleinfer{}{
    \Gamma \types V : \typrod{\mathsf{Body}}{\typrod{\mathsf{Door}}{\mathsf{Door}}} \\
    \Gamma, \ctxmod {\tau_{\text{paint}}}, 
      x \of \typrod{\tybox{\tau_{\text{dry}}}{\mathsf{Body}}}{\typrod{\tybox{\tau_{\text{dry}}}{\mathsf{Door}}}{\tybox{\tau_{\text{dry}}}{\mathsf{Door}}}} 
        \types M : \tycomp{X}{\tau}
  }{
    \Gamma \types \tmop{paint}{V}{\tmcont x M} : \tycomp{X}{\tau_{\text{paint}} + \tau}
  }
\]
Here, $\ctxmod {\tau_{\text{paint}}}$ expresses that from the perspective of any $\tmkw{unbox}$es
in $M$, an \emph{additional $\tau_{\text{paint}}$ time} will have passed compared to the 
beginning of the execution of $\tmop{paint}{V}{\tmcont x M}$, which is
typed in the ``earlier'' context $\Gamma$. Also,  
observe that $\tm{paint}$'s result $x$ is available \emph{after} $\tau_{\text{paint}}$ time has passed 
(i.e., after $\tm{paint}$ finishes), and its type 
has the car part types wrapped as temporal resources, 
ensuring that any further operations (e.g., $\tm{assemble}$) can access them
only after \emph{at least} $\tau_{\text{dry}}$ time has passed \emph{after} $\tm{paint}$ finishes. 
The $\tm{delay}\; \tau$ operation is typed 
analogously.

Finally, similarly to algebraic operations, we also use the context modality $\ctxmod {\tau}$ 
to model the passage of time in sequential composition, as specified in
\[
  \ruleinfer{}{
    \Gamma \types M : \tycomp{X}{\tau} \\
    \Gamma, \ctxmod \tau, x \of X \types N : \tycomp{Y}{\tau'}
  }{
    \Gamma \types \tmlet{x}{M}{N} : \tycomp{Y}{\tau + \tau'}
  }
\]
The type $\tycomp{X}{\tau}$ (for specifying 
the execution time of computations) is standard from 
graded monads style effect systems~\cite{Katsumata:GradedMonads}. 
The novelty of our work is to 
use this effect information to inform continuations that they can 
safely assume that the given amount of additional time has passed 
before they start executing.

\subsubsection{Putting It All Together}

We conclude this overview by revisiting the production line code snippet and note that
in the $\lambdaTau$-calculus we can write it as
\[
\small
\begin{array}{l}
\tmlet{(\text{body'},\text{left-door'},\text{right-door'})}{\tmgeneff{paint}{(\text{body},\text{left-door},\text{right-door})}}{
  \\
  \tm{delay}\; \tau_{\text{dry}};
  \\
  \tmunbox{\text{body'}}{\text{body''}}{}
  \\
  \tmunbox{\text{left-door'}}{\text{left-door''}}{}
  \\
  \tmunbox{\text{right-door'}}{\text{right-door''}}{}
  \\
  \tmgeneff{assemble}{(\text{body''},\text{left-door''},\text{right-door''})}}
\end{array}
\]

Observe that apart from the $\tmkw{unbox}$ operations, the code looks identical 
to the naive, unsafe solution discussed earlier. However, crucially, now any code that wants to
use the outputs of $\tm{paint}$ will typecheck only if these resources are accessed after
at least $\tau_{\text{dry}}$ time units have passed after $\tm{paint}$ finishes. In the code snippet, 
this is achieved by blocking execution
with $\tm{delay}\; \tau_{\text{dry}}$ for $\tau_{\text{dry}}$ time units, but this could have been 
equally well achieved by executing other useful operations $\op_1 ;\; \ldots ;\; \op_n$, 
as long as they collectively take at least $\tau_{\text{dry}}$ time.

% !TEX root = temporal-resources.tex

\section{A Calculus for Programming with Temporal Resources}
\label{sect:core-calculus}

We now recast the ideas explained above 
as a formal, modally typed, effectful core calculus, called $\lambdaTau$. 
We base it on the fine-grain call-by-value $\lambda$-calculus~\cite{Levy:FGCBV}.
 
\subsection{Types}

The types of $\lambdaTau$ are given in \cref{fig:types}. \emph{Ground types}
include base types $\tybase$, and are closed under finite products and the modal \emph{temporal 
resource type} $\tybox{\tau}{A}$. The latter denotes that an $A$-typed value 
will become available in \emph{at most} $\tau$ time units, where $\tau \in \mathbb{N}$ counts 
discrete time moments.\footnote{For concreteness, we work with 
  $(\mathbb{N}, 0, +, \dotminus, \le)$ for time grades, but we do not 
  foresee problems generalising these to come from other analogous 
  algebraic structures.} 
The ground types can also come with \emph{constants} $\tmconst{f}$ with 
associated \emph{constant signatures} $\tmconst{f} : (A_1,\ldots,A_n) \to B$.

To model operations such as $\tm{paint}$ and $\tm{assemble}$ discussed 
in \cref{sect:overview-temporal-resources}, we assume a set of \emph{operations symbols} $\Ops$, 
with each $\op \in \Ops$ assigned an \emph{operation signature} 
$\op : \tysigop{A_\op}{B_\op}{\tau_\op}$, which specifies that $\op$ accepts inputs of
type $A_\op$, returns values of type $B_\op$, and its execution takes $\tau_\op$ time units.
Observe that by typing operations with ground types, as opposed to simply with base types, 
we can specify operations such as 
$\tm{paint} : \tysigop{\mathsf{Part}}{(\tybox{\tau_{\text{dry}}}{\mathsf{Part}})}{\tau_{\text{paint}}}$,
returning values that can be accessed only after a certain amount of time, here, after $\tau_{\text{dry}}$.

\emph{Value types} extend ground types with \emph{function type} $\tyfun{X}{\tycomp{Y}{\tau}}$ 
that specifies functions taking $X$-typed arguments to computations 
that return $Y$-typed values and take $\tau$ time to execute, as expressed 
by the \emph{computation type} $\tycomp{Y}{\tau}$.

\begin{figure}[tb]
  \parbox{\textwidth}{
  \centering
  \begin{align*}
  \text{Time grade:}
  \phantom{\bnfis}& \tau \in \mathbb{N}
  \\[0.5ex]
  \text{Ground type $A$, $B$, $C$}
  \bnfis& \tybase \bnfor \!\!\tyunit \bnfor \!\!\typrod{A}{B} \bnfor \!\!\tybox{\tau}{A} 
  \\[0.5ex]
  \text{Value type $X$, $Y$, $Z$}
  \bnfis& A \bnfor \!\!\typrod{X}{Y} \bnfor \!\!\tyfun{X}{\tycomp{Y}{\tau}} \bnfor \!\!\tybox{\tau}{X}
  \\[0.5ex]
  \text{Computation type:}
  \phantom{\bnfis}& \tycomp{X}{\tau}
  \end{align*}
  } 
  \caption{Types of $\lambdaTau$.}
  \label{fig:types}
\end{figure}

\subsection{Terms}

The syntax of terms is given in \cref{fig:terms}, separated into values and computations.

\begin{figure}[h]
  \parbox{\textwidth}{
  \centering
  \abovedisplayskip=0pt
  \begin{align*}
  \intertext{\textbf{Values}}
  V, W
  \bnfis& x                                       & &\text{variable} \\
  \bnfor& \tmconst{f}(V_1,\ldots,V_n)               & &\text{constant} \\
  \bnfor& \tmunit \bnfor  \!\!\tmpair{V}{W}                         & &\text{unit and pairing} \\
  \bnfor& \tmfun{x : X}{M}                        & &\text{function} \\
  \bnfor& \tmbox[\tau]{V}                        & &\text{boxing up a temporal resource}
  \\[-0.5ex]
  \intertext{\textbf{Computations}}
  M, N
  \bnfis& \tmreturn{V}                            & &\text{returning a value} \\
  \bnfor& \tmlet{x}{M}{N}                           & &\text{sequential composition} \\
  \bnfor& V\,W                                    & &\text{function application} \\
  \bnfor& \tmmatch{V}{\tmpair{x}{y} \mapsto N}    & &\text{pattern-matching} \\
  \bnfor& \tmop{op}{V}{\tmcont x M}       & &\text{operation call} \\
  \bnfor& \tmdelay{\tau}{M}       & &\text{time delay} \\
  \bnfor& \tmhandle{M}{H}{x}{N}       & &\text{effect handling} \\
  \bnfor& \tmunbox[\tau]{V}{x}{N}       & &\text{unboxing a temporal resource}
  \\[-0.5ex]
  \intertext{\textbf{Effect handlers}}
  H
  \bnfis& {\tmopclause{x}{k}{M_{\op}}}_{\op \in \Ops}                           & &\text{operation clauses}
  \end{align*}
  } 
  \caption{Values, computations, and effect handlers of $\lambdaTau$.}
  \label{fig:terms}
\end{figure}

\emph{Values} include variables, constants, finite tuples, functions, and the \emph{boxing up of 
temporal resources}, $\tmbox[\tau]{V}$, 
which allows us to consider an arbitrary value $V$ as a temporal resource 
as long as it is safe to access $V$ after $\tau$ time units.

\emph{Computations} include returning values, sequential composition, function 
application, pattern-matching\footnote{The form $\tmlet{(x,y,z)}{M}{N}$ in \cref{sect:introduction},\ref{sect:overview} is the natural combination of $\tmkw{let}$ and $\tmkw{match}$.}, algebraic operation calls, effect handling, and the 
\emph{unboxing of temporal resources}, where given a temporal resource $V$ of type $\tybox{\tau}{X}$, 
the computation $\tmunbox[\tau]{V}{x}{N}$ is used to access the underlying value of 
type $X$ if at least $\tau$ time units have passed since the creation of the resource $V$.

In addition to user-specifiable operation calls (via operation signatures and effect handling), 
we include a separate $\tmdelay{\tau}{M}$ operation that blocks the execution 
of its continuation for the given amount of time. 
For simplicity, we require effect handlers 
to have \emph{operation clauses} $M_\op$ for all $\op \in \Ops$, but
we do not allow
$\tmkw{delay}$s to be handled in light of the equations we want of them in \cref{sect:delay-equations}, 
where all consecutive $\tmkw{delay}$s are collapsed and all zero-$\tmkw{delay}$s are removed.

\subsection{Type System}
\label{sect:type-system}

We now equip $\lambdaTau$ with a modal type-and-effect system. On 
the one hand, for modelling temporal resources, we build on Fitch-style 
modal type systems~\cite{Clouston:FitchStyle}. On the other hand, for
modelling effectful computations and their specifications, we build on 
type-and-effect systems for calculi based on graded 
monads~\cite{Katsumata:GradedMonads}.

The \emph{typing judgements} are 
written as $\Gamma \types V : X$ and 
$\Gamma \types M : \tycomp{X}{\tau}$, where $\tau$ specifies $M$'s
execution time and $\Gamma$ is a \emph{temporal typing context}, given by
\[
\Gamma \bnfis \cdot \bnfor \Gamma, x \of X \bnfor \Gamma, \ctxmod \tau
\]
Here, $\ctxmod \tau$ is a \emph{temporal context modality}, akin to 
$\ctxunlock$ in Fitch-style systems. We use it to
express that when typechecking a term of the form $\Gamma, \ctxmod \tau \types V : X$, 
we can safely assume that \emph{at least} $\tau$ time will have passed before 
the resource $V$ is accessed or executed. The \emph{rules} defining these 
judgements are given in \cref{fig:typing-rules}.

\begin{figure}[tp]
  \centering
  
  % VALUES
  
  \textbf{Values}
  \begin{mathpar}
  \ruleinfer{Var}{
    \phantom{...}
  }{
    \Gamma, x \of X, \Gamma' \types x : X
  }

  \ruleinfer{Const}{
    (\Gamma \types V_i : A_i)_{1 \leq i \leq n}
  }{
    \Gamma \types \tmconst{f}(V_1, \ldots , V_n) : B
  }

  \ruleinfer{Pair}{
    \Gamma \types V : X \\
    \Gamma \types W : Y
  }{
    \Gamma \types \tmpair{V}{W} : \typrod{X}{Y}
  }
  \end{mathpar}
  \begin{mathpar}
  \ruleinfer{Unit}{
  }{
    \Gamma \types \tmunit : \tyunit
  }

  \ruleinfer{Fun}{
    \Gamma, x \of X \types M : \tycomp{Y}{\tau}
  }{
    \Gamma \types \tmfun{x : X}{M} : \tyfun{X}{\tycomp{Y}{\tau}}
  }

  \ruleinfer{Box}{
    \Gamma, \ctxmod \tau \types V : X
  }{
    \Gamma \types \tmbox[\tau]{V} : \tybox{\tau}{X}
  }
  \end{mathpar}
  
  \vspace{1.25ex}
  
  % COMPUTATIONS
    
  \textbf{Computations}
  \begin{mathpar}
  \ruleinfer{Return}{
    \Gamma \types V : X
  }{
    \Gamma \types \tmreturn{V} : \tycomp{X}{0}
  }
  
  \ruleinfer{Let}{
    \Gamma \types M : \tycomp{X}{\tau} \\
    \Gamma, \ctxmod \tau, x \of X \types N : \tycomp{Y}{\tau'}
  }{
    \Gamma \types \tmlet{x}{M}{N} : \tycomp{Y}{\tau + \tau'}
  }
  \end{mathpar}
  \begin{mathpar}
  \ruleinfer{Apply}{
    \Gamma \types V : \tyfun{X}{\tycomp{Y}{\tau}} \\
    \Gamma \types W : X
  }{
    \Gamma \types \tmapp{V}{W} : \tycomp{Y}{\tau}
  }
  
  \ruleinfer{Match}{
    \Gamma \types V : \typrod{X}{Y} \\
    \Gamma, x \of X, y \of Y \types N : \tycomp{Z}{\tau}
  }{
    \Gamma \types \tmmatch{V}{\tmpair{x}{y} \mapsto N} : \tycomp{Z}{\tau}
  }
  \end{mathpar}
  \begin{mathpar}
  \ruleinfer{Op}{
    \Gamma \types V : A_\op \\
    \Gamma, \ctxmod {\tau_\op}, x \of B_\op \types M : \tycomp{X}{\tau}
  }{
    \Gamma \types \tmop{op}{V}{\tmcont x M} : \tycomp{X}{\tau_\op + \tau}
  }
  
  \ruleinfer{Delay}{
    \Gamma, \ctxmod {\tau} \types M : \tycomp{X}{\tau'}
  }{
    \Gamma \types \tmdelay{\tau}{M} : \tycomp{X}{\tau + \tau'}
  }
  \end{mathpar}
  \vspace{-1ex}
  \begin{mathpar}
  \ruleinfer{Handle}{
    \Gamma \types M : \tycomp{X}{\tau} \\
    \Gamma, \ctxmod {\tau}, x \of X \types N : \tycomp{Y}{\tau'} \\
    H = {\tmopclause{x}{k}{M_{\op}}}_{\op \in \Ops} \\
    \big( \forall \tau'' \;.\; \Gamma , x \of A_\op , k \of \tybox{\tau_\op}{(\tyfun{B_\op}{\tycomp{Y}{\tau''}})} \types M_\op : \tycomp{Y}{\tau_\op + \tau''} \big)_{\op \in \Ops}\\
  }{
    \Gamma \types \tmhandle{M}{H}{x}{N}: \tycomp{Y}{\tau + \tau'}
  }
  \end{mathpar}
  \begin{mathpar}
  \ruleinfer{Unbox}{
    \tau \le \ctxtime \Gamma \\
    \Gamma \ctxminus \tau \types V : \tybox{\tau}{X} \\
    \Gamma, x \of X \types N : \tycomp{Y}{\tau'}
  }{
    \Gamma \types \tmunbox[\tau]{V}{x}{N} : \tycomp{Y}{\tau'}
  }
  
  \end{mathpar}
  
  \caption{Typing rules of $\lambdaTau$.}
  \label{fig:typing-rules}
\end{figure}

In contrast to Fitch-style modal type systems discussed in 
\cref{sect:overview-modal-types}, \rulename{Var} does not restrict the $\Gamma'$  
right of $x$ to not include any context modalities. This is so because 
in the possible worlds reading of $\lambdaTau$ (see \cref{sect:semantics})
we treat all types as being monotone for time---this is not usually 
the case for formulae in modal logics such as S4, but  in $\lambdaTau$ 
this models that once any value is available it will remain so.

As in systems based on graded monads,  
\rulename{Return} specifies that returning a value takes zero time, 
and \rulename{Let} 
that the execution time of sequentially composed computations 
is the sum of the individual ones. Novel to $\lambdaTau$, \rulename{Let}, \rulename{Op}, 
\rulename{Delay}, and \rulename{Handle} state that the continuations 
can safely
assume that relevant amount of additional time has
passed before they start executing, as discussed in \cref{sect:overview-temporal-resources}.

When typing the operation clauses $M_\op$ in \rulename{Handle}, 
we universally quantify (at the meta-level) over the execution time $\tau''$ 
of the continuation $k$ of $M_\op$. We do so as 
the operation clauses $M_\op$ must be able to execute at any point
when effect handling recursively traverses $M$. Further, observe that 
$k$ is wrapped inside a resource type. This ensures that 
$k$ is invoked only after $\tau_\op$ amount of time has been spent in 
$M_\op$, thus guaranteeing that the temporal discipline is respected. 
Note that this enforces a \emph{linear} discipline for our
effect handlers: for $\tau_\op > 0$, $k$ must be executed 
exactly once for $M_\op$'s execution time to match $\tau_\op + \tau''$.

Finally, \rulename{Box} specifies that in order to box up a value $V$ of type $X$ 
as a temporal resource of type $\tybox{\tau}X$, we must be able to type 
$V$ when assuming that $\tau$ additional time units will have passed
before $V$ is accessed. At the same time, \rulename{Unbox} specifies that we can 
unbox a temporal resource $V$ of type $\tybox{\tau}X$ only if at least $\tau$
time units have passed since its creation: the time captured by $\Gamma$ must
be at least $\tau$, and we must be able to type $V$ in a $\tau$ time
units ``earlier'' context $\Gamma \ctxminus \tau$. 
The \emph{time captured by a context}, $\ctxtime \Gamma$, is calculated recursively as 
\[
\begin{array}{c}
\ctxtime \cdot \defeq 0
\qquad
\ctxtime {(\Gamma, x \of X)} \defeq \ctxtime \Gamma
\qquad
\ctxtime {(\Gamma, \ctxmod {\tau})} \defeq \ctxtime \Gamma + \tau
\end{array}
\]
and the \emph{``time travelling'' operation} $\Gamma \ctxminus \tau$ as
(where $\tau_+ \equiv 1 + \tau''$ for some $\tau''$)
\[
\begin{array}{c}
  \Gamma \ctxminus 0 \defeq \Gamma
  \qquad
  \cdot \ctxminus\, \tau_+ \defeq \cdot
  \qquad
  (\Gamma, x \of X) \ctxminus \tau_+ \defeq \Gamma \ctxminus \tau_+
\\[1ex]
  (\Gamma, \ctxmod {\tau'}) \ctxminus \tau_+ \defeq 
    \cond
      {\tau_+ \le \tau'}
      {\Gamma, \ctxmod {\tau' \dotminus \tau_+}}
      {\Gamma \ctxminus (\tau_+ \dotminus \tau')}
\end{array}
\]
taking $\Gamma$ to an 
``earlier'' state by removing $\tau$ worth of modalities and
variables.

\subsection{Admissibility of Renamings and Substitutions}

We now show that expected \emph{structural}
and \emph{substitution rules}~\cite{Barendregt:LCWithTypes} 
are admissible.
  
\begin{theorem}
\label{thm:renaming}
The typing relations $\Gamma \types V : X$ and 
$\Gamma \types M : \tycomp{X}{\tau}$ are closed under standard structural 
rules of weakening, exchange of consecutive variables, and contraction 
(omitted here). Furthermore, both typing relations are also closed under
rules making $\ctxmod -$ into a strong monoidal functor (with a co-strength)~\cite{MacLane:CatWM}: 
\[
\begin{prooftree}[rule style=double]
  \Hypo{{\Gamma, \ctxmod {0}} \types J}
  \Infer1{\Gamma \types J}
\end{prooftree}
\qquad
\begin{prooftree}[rule style=double]
  \Hypo{\Gamma, \ctxmod {\tau_1 + \tau_2} \types J}
  \Infer1{\Gamma, \ctxmod {\tau_1}, \ctxmod {\tau_2} \types J}
\end{prooftree}
\qquad
\begin{prooftree}
  \Hypo{{\Gamma, \ctxmod {\tau}} \types J \quad \tau \le \tau'}
  \Infer1{\Gamma, \ctxmod {\tau'} \types J}
\end{prooftree}
\qquad
\begin{prooftree}
  \Hypo{{\Gamma, \ctxmod {\tau}, x \of X} \types J}
  \Infer1{\Gamma, x \of X, \ctxmod {\tau} \types J}
\end{prooftree}
\]
where $\Gamma \types J$ ranges over both typing relations, where 
the first two rules hold in both directions, and the last rule
expresses that if we can type $J$ using a variable ``now'', we can also 
type $J$ if that variable was brought into scope ``earlier''.
\end{theorem}

\begin{proof}
First, we define a \emph{renaming relation} $\rho : \Gamma \leadsto \Gamma'$, 
and then prove by induction that if  $\Gamma \types J$ and $\rho : \Gamma \leadsto \Gamma'$ 
then $\Gamma' \types J[\rho]$, where $J[\rho]$ is $J$ renamed with $\rho$.
The $\leadsto$ relation is defined as the reflexive-transitive-congruent 
closure of rules corresponding to the desired structural rules, e.g., 
$\mathsf{var}^r_{x \of X \in \Gamma} : \Gamma, y \of X \leadsto \Gamma$ and
$\mu^r : \Gamma, \ctxmod {\tau_1 + \tau_2} \leadsto \Gamma, \ctxmod {\tau_1}, \ctxmod {\tau_2}$. 
The full list is given in the online appendix.

For the \rulename{Var} and \rulename{Unbox} cases of the proof, we show that if 
$\rho : \Gamma \leadsto \Gamma'$ and $x \in_\tau \Gamma$, then 
$\rho\, x \in_{\tau'} \Gamma'$ for some $\tau'$ with $\tau \le \tau'$, where 
$x \in_\tau \Gamma$ means that $x \in \Gamma$ and there is $\tau$ worth of 
modalities right of $x$ in $\Gamma$, and $\rho\, x$ is the variable that $\rho$ 
maps $x$ to. For \rulename{Unbox}, we further prove that
if $\rho : \Gamma \leadsto \Gamma'$, then for any $\tau$ we can build  
$\rho \ctxminus \tau : \Gamma \ctxminus \tau \leadsto \Gamma' \ctxminus \tau$, 
using the result about $\in_{\tau}$ to ensure that $\rho$ does not map
any $x \in \Gamma \ctxminus \tau$ outside of $\Gamma' \ctxminus \tau$.
We also establish that if $\Gamma \leadsto \Gamma'$, 
then $\ctxtime \Gamma \le \ctxtime {\Gamma'}$, allowing us to deduce $\tau \le \ctxtime{\Gamma'}$ 
from $\tau \le \ctxtime \Gamma$.
\end{proof}

The admissibility of the rules corresponding to $\mu^r$ 
(and its inverse) relies on us having defined context splitting in \rulename{Unbox}
using $\Gamma \ctxminus \tau$, as opposed to more rigidly as $\Gamma,\Gamma'$, 
as in~\cite{Clouston:FitchStyle}, as then it would be problematic if the split happens between
$\ctxmod{\tau_1},\ctxmod{\tau_2}$.
Inverses of the last two rules in \cref{thm:renaming}
are not valid---they would allow unboxing temporal resources without enough time having passed.

\begin{theorem}
The typing relations $\Gamma \types V : X$ and $\Gamma \types M : \tycomp{X}{\tau}$ are closed 
under substitution, i.e., if $\Gamma, x \of X, \Gamma' \types J$ and 
$\Gamma \types W : X$, then $\Gamma, \Gamma' \types J[W/x]$, where $J[W/x]$ is standard   
recursively defined capture-avoiding substitution~\cite{Barendregt:LCWithTypes}.
\end{theorem}

\begin{proof}
The proof proceeds by induction on the derivation of $\Gamma, x \of X, \Gamma' \types J$.
The most involved case is \rulename{Unbox}, where we construct the derivation of
$\Gamma,\Gamma' \types \tmunbox[\tau]{V[W/x]}{y}{N[W/x]} : \tycomp{Y}{\tau'}$ by first analysing 
whether $\tau \le \ctxtime \Gamma'$, which tells us whether $x$ is in the context 
$(\Gamma, x \of X, \Gamma') \ctxminus \tau$ of $V$, based on which 
we learn whether $W$ continues to be substituted for $x$ in $V$ or whether $V[W/x] = V$.
\end{proof}

\subsection{Equational Theory}
\label{sect:equational-theory}

We conclude the definition of $\lambdaTau$ by equipping it with an \emph{equational theory} 
to reason about program equivalence, defined
using judgements $\Gamma \types V \equiv W : X$ 
and $\Gamma \types M \equiv N : \tycomp{X}{\tau}$, where we presuppose that the 
terms are well-typed for the given contexts and types. The rules defining these
relations are given in \cref{fig:equations}. We omit standard equivalence, congruence, 
and substitutivity rules~\cite{Barendregt:LCWithTypes}. 

\begin{figure}[h]
  \centering  
  \parbox{\textwidth}{
  \mathtoolsset{original-shortintertext=false,below-shortintertext-sep=0pt,above-shortintertext-sep=0pt}
  \begin{align*}
    \\
    \tmunit &\equiv V : \tyunit & (\eta)
    \\
    \tmfun{x : X}{\tmapp{V}{x}} &\equiv V : \tyfun{X}{\tycomp{Y}{\tau}} & (\eta)
    \\[1.25ex]
    \tmapp{(\tmfun{x \of X}{M})}{V} &\equiv M[V/x] & (\beta)
    \\[1.25ex]
    \tmmatch{\tmpair{V}{W}}{\tmpair{x}{y} \mapsto N} &\equiv N[V/x,W/y] & (\beta)
    \\
    \tmmatch{V}{\tmpair{x}{y} \mapsto N[\tmpair{x}{y}/z]} &\equiv N[V/z] & (\eta)
    \\[1.25ex]
    \tmlet{x}{(\tmreturn{V})}{N} &\equiv N[V/x] & (\beta)
    \\
    \tmlet{y}{(\tmlet{x}{M}{N})}{P} &\equiv \tmlet{x}{M}{(\tmlet{y}{N}{P})} & (\beta)
    \\
    \tmlet{x}{M}{\tmreturn{x}} &\equiv M & (\eta)
    \\[1.25ex]
    \tmlet{x}{(\tmop{op}{V}{\tmcont y M})}{N} &\equiv \tmop{op}{V}{\tmcont y {\tmlet x M N}} & (\beta)
    \\
    \tmlet{x}{(\tmdelay{\tau}{M})}{N} &\equiv \tmdelay{\tau}{(\tmlet x M N)} & (\beta)
    \\[1.25ex]
    \tmhandle{(\tmreturn{V})}{H}{x}{N} &\equiv N[V/x] & (\beta)
    \\[0.5ex]
    \tmhandle{(\tmop{op}{V}{\tmcont y M})}{H}{x}{N} &\equiv \\
    M_\op[V/x,\tmbox[\tau_\op]{(\tmfun{y : B_\op}{&\; \tmhandle{M}{H}{x}{N}})}/k] & (\beta)
    \\[0.5ex]
    \tmhandle{(\tmdelay{\tau}{M})}{H}{x}{N} &\equiv \tmdelay{\tau}{(\tmhandle{M}{H}{x}{N})} & (\beta)
    \\[1.25ex]
    \tmunbox[\tau]{(\tmbox[\tau]{V})}{x}{N} &\equiv N[V/x] & (\beta)
    \\
    \tmunbox[\tau]{V}{x}{N[\tmbox[\tau]{x}/y]} &\equiv N[V/y] & (\eta)
  \end{align*}
  } % parbox

  \caption{Equational theory of $\lambdaTau$.}
  \label{fig:equations}
\end{figure}

The equational theory consists of standard $\beta/\eta$-equations for the unit, product, 
and function types. We also include monadic equations for $\tmkw{return}$ and 
$\tmkw{let}$~\cite{Moggi:ComputationalLambdaCalculus}. For $\op$ and $\tmkw{delay}$, we include  
algebraicity equations, allowing us to pull them out of $\tmkw{let}$~\cite{Bauer:WhatIsAlgebraic}. 
For $\tmkw{handle}$, we include equations expressing that
effect handling recursively traverses a term, replacing each $\op$-occurrence with the 
operation clause $M_\op$, leaving $\tmkw{delay}$s untouched, and finally 
executes the continuation $N$ when reaching return values~\cite{Plotkin:HandlingEffects}. Finally, we include 
$\beta$/$\eta$-equations for $\tmkw{box}$ and $\tmkw{unbox}$,   
expressing that $\tmkw{unbox}$ behaves as a pattern-matching elimination form for $\tmkw{box}$.

% !TEX root = temporal-resources.tex

\section{Denotational Semantics}
\label{sect:semantics}

We justify the design of $\lambdaTau$ by giving it a mathematically natural semantics based on  
\emph{adjunctions between strong monoidal 
functors}~\cite{MacLane:CatWM} (modelling modalities) and a \emph{strong\footnote{To be more specific, we use a modal notion 
of \emph{$\futuremod -$-strength} that we define below.} graded monad}~\cite{Katsumata:GradedMonads} 
(modelling computations).
We assume general knowledge of category theory, only spelling out details
specific to $\lambdaTau$. To optimise for space, 
we discuss the abstract model structure simultaneously with a concrete example 
using presheaves~\cite{MacLane:Sheaves}, but note that the interpretation is defined, 
and its soundness proved, with respect to the abstract structure.

When referring to the \emph{abstract model structure}, we denote the underlying category with $\Cat$.
Meanwhile, the \emph{concrete presheaf example} is given in $\Pshf$, consisting of  
functors from $(\mathbb{N},\le)$ 
to the category $\Set$ of sets and functions. 

The model in $\Pshf$ is similar to Kripke's possible worlds semantics, 
except that in $\Pshf$ all objects are \emph{monotone} for $\le$, i.e., 
for any $A \in \Pshf$ we have functions $A(t_1 \le t_2) : A(t_1) \to A(t_2)$ 
respecting reflexivity and transitivity,
whereas Kripke models are commonly given by discretely indexed presheaves and only 
modalities change worlds.
For $\lambdaTau$, working in $\Pshf$ gives us that when a resource
becomes available, it will remain so without need for reboxing, leading 
to a more natural system for temporal resources and a simpler \rulename{Var} rule.

\subsection{Interpretation of Types}
\label{sect:interpretation-types}

\subsubsection{Value Types and Contexts}

To interpret value types, we require the category $\Cat$ to have \emph{finite products} 
$(\One, A \times B)$ and \emph{exponentials} $A \expto B$, 
so as to model the unit, product, and function types. In $\Pshf$, the former are 
given point-wise using the finite products in $\Set$, and the 
latter are given as $(A \expto B)(t) \defeq \Pshf(\mathsf{hom}\, t \times A, B)$, 
where $\mathsf{hom}\, t : (\mathbb{N},\le) \to \Set$ is the covariant 
\emph{hom-functor} for $(\mathbb{N},\le)$, given by $\mathsf{hom}\, t \defeq t \le (-)$~\cite{MacLane:Sheaves}.
When unfolding it further, the above means that  
$(A \expto B)(t)$ is the set of functions $(f_{t'} : A(t') \to B(t'))_{t' \in \{t' \in \mathbb{N} \vert t \le t'\}}$
that are natural in $t'$, capturing the intuition that in $\lambdaTau$ functions
can be applied in any future context. For base types, we require 
an object $\sem{\tybase}$ of $\Cat$ for each $\tybase$.

To interpret the temporal resource type, we require  
a \emph{strong monoidal functor} $\futuremod - : (\mathbb{N},\le) \to [ \Cat , \Cat ]$, where
$[ \Cat , \Cat ]$ is the category of endofunctors on $\Cat$.
This means that we have
functors $\futuremod \tau : \Cat \to \Cat$, for all $\tau \in \mathbb{N}$, together with 
morphisms $\futuremod {\tau_1 \le \tau_2}_A : \futuremod{\tau_1} A \to \futuremod{\tau_2}  A$, 
natural in $A$ and respecting $\le$.
Strong monoidality of $\futuremod -$ means that we have natural isomorphisms
${\varepsilon_A : \futuremod 0 A \overset{\cong}{\to} A}$ and $\delta_{A,\tau_1,\tau_2} : 
\futuremod {\tau_1 + \tau_2} A \overset{\cong}{\to} \futuremod {\tau_1} (\futuremod {\tau_2} A)$, satisfying 
time-graded variants of comonad laws~\cite{Beck:Triples}:
\[
\varepsilon \circ \delta_{A,0,\tau} \equiv \id
\qquad
\futuremod \tau (\varepsilon) \circ \delta_{A,\tau,0} \equiv \id
\qquad
\delta_{\futuremod {\tau_3} A,\tau_1,\tau_2} \circ \delta_{A,\tau_1 + \tau_2,\tau_3} 
  \equiv \futuremod {\tau_1} (\delta) \circ \delta
\]
We also require $(\delta_{A,\tau_1,\tau_2},\delta^{-1}_{A,\tau_1,\tau_2}) $ to be monotone in 
$\tau_1, \tau_2$, i.e., if $\tau_1 \le \tau'_1$ and $\tau_2 \le \tau'_2$, then 
$\futuremod {\tau'_1} (\futuremod {\tau_2 \le \tau'_2 }) \circ \futuremod {\tau_1 \le \tau'_1} \circ \delta
  \equiv \delta \circ \futuremod{\tau_1 + \tau_2 \le \tau'_1 + \tau'_2}_A$. 
We omit the indices of the components of natural transformations when convenient.

In $\Pshf$, we define $(\futuremod \tau A)(t) \defeq A(t + \tau)$, 
with $\futuremod \tau A$-values given by future $A$-values, 
and with $(\varepsilon_A,\varepsilon_A^{-1},\delta_A,\delta_A^{-1})$ 
given by identities on $A$-values, combined with the laws of $(0,+)$, e.g., as 
$(\varepsilon_A)_t\, \big(a \in (\futuremod{0}A)(t) \equiv A(t + 0)\big) \defeq a \in A(t)$.

Using the above, we interpret a value type $X$ as an object $\sem{X}$ of $\Cat$, as
\[
\begin{array}{c}
  \sem{A} \defeq \gsem{A}
  \qquad
  \sem{\tyunit} \defeq \One
  \qquad
  \sem{\typrod{X}{Y}} \defeq \sem{X} \times \sem{Y}
  \\[1ex]
  \sem{\tyfun{X}{\tycomp{Y}{\tau}}} \defeq \sem{X} \expto \T\, \tau \, \sem{Y}
  \qquad
  \sem{\tybox{\tau}{X}} \defeq \futuremod \tau \sem{X}
\end{array}
\]
where $\T$ is a graded monad for modelling computations---we 
return to it below. The interpretation of ground types $\gsem{A}$
is defined similarly, so we omit it here.

Next, we define the interpretation of contexts, for which we require another \emph{strong 
monoidal functor}, $\pastmod - : (\mathbb{N},\le)^{\text{op}} \to [ \Cat , \Cat ]$. 
Note that $\pastmod -$ is \emph{contravariant}---this enables us to model the  
structural rules that allow terms typed in an earlier context 
to be used in future ones (see \cref{thm:renaming}). We denote the strong 
monoidal structure of $\pastmod -$ with $\eta_A : A \overset{\cong}{\to} \pastmod 0 A$ and 
$\mu_{A, \tau_1, \tau_2} : \pastmod {\tau_1} (\pastmod {\tau_2} A) 
  \overset{\cong}{\to} \pastmod {\tau_1 + \tau_2} A$, 
required to satisfy time-graded variants of monad laws~\cite{MacLane:CatWM}, given by
\[
\mu_{\!A,0,\tau} \circ \eta \equiv \id
\qquad
\mu_{\!A,\tau,0} \circ \pastmod{\tau}(\eta) \equiv \id
\qquad
\mu_{A,\tau_1 + \tau_2, \tau_3} \circ \mu_{\pastmod{\tau_3}A,\tau_1,\tau_2} \equiv \mu \circ \pastmod{\tau_1}(\mu)
\]
and $(\mu_{A, \tau_1, \tau_2},\mu^{-1}_{A, \tau_1, \tau_2})$ have to be monotone in 
$\tau_1,\tau_2$, similarly to $(\delta,\delta^{-1}) $ above.

In $\Pshf$, we define  
$(\pastmod \tau A)(t) \defeq (\tau \le t) \times A (t \dotminus \tau)$, as past $A$-values, with the 
\emph{side-condition} $\tau \le t$ crucial for the existence of the adjunctions $\pastmod \tau \dashv \futuremod \tau$ 
we require below. We define $(\eta_A,\eta_A^{-1},\mu_A,\mu_A^{-1})$  
similarly to earlier, as identities on $A$-values, combined with the laws of $(0,+,\dotminus)$, 
so as to satisfy the side-conditions. 

With this, we can interpret \emph{contexts} $\Gamma$ as \emph{functors} $\sem{\Gamma} : \Cat \to \Cat$, given by:
\[
\begin{array}{c}
  \sem{\cdot} A \defeq A
  \qquad
  \sem{\Gamma, x \of X} A \defeq \sem{\Gamma} A \times \sem{X}
  \qquad
  \sem{\Gamma, \ctxmod \tau} A \defeq \pastmod \tau (\sem{\Gamma} A)
\end{array}
\]
We interpret contexts as functors to easily manipulate denotations of 
composite contexts, e.g., we then have 
$\iota_{\Gamma;\Gamma';A} : \sem{\Gamma,\Gamma'}A \overset{\cong}{\to} \sem{\Gamma'}(\sem{\Gamma}A)$, 
natural in $A$.

Finally, to formulate the semantics of computation types and terms, we  
require there to be a family of \emph{adjunctions} $\pastmod \tau \dashv \futuremod \tau$, 
i.e., natural transformations $\etaA_{A,\tau} : A \to \futuremod{\tau} (\pastmod{\tau} A)$ (the \emph{unit})
and $\epsA_{A,\tau} : \pastmod{\tau} (\futuremod{\tau} A) \to A$ (the \emph{counit}), for all $\tau \in \mathbb{N}$, satisfying time-graded
variants of standard adjunction laws~\cite{MacLane:CatWM}, given by
\[
\epsA_{\pastmod{\tau}A,\tau} \circ \pastmod{\tau}(\etaA_{A,\tau}) \equiv \id
\qquad
\futuremod{\tau}(\epsA_{A,\tau}) \circ \etaA_{\futuremod{\tau}A,\tau} \equiv \id
\]
We also require 
$(\etaA,\epsA)$ to interact well with the strong monoidal structures:
\[
\hspace{-0.1cm}
\small
\begin{array}{c}
\futuremod {\tau} (\pastmod {0 \le \tau}) \circ \etaA_{A,\tau} \circ \eta^{-1} \circ \varepsilon \equiv \futuremod {0 \le \tau}
\qquad
\futuremod {\tau_1} (\futuremod {\tau_2} (\mu)) \circ \futuremod{\tau_1}(\etaA_{\pastmod{\tau_1}A,\tau_2}) \circ \etaA_{A,\tau_1}
  \equiv \delta \circ \etaA
\\[1.25ex]
\pastmod{0}(\futuremod{0 \le \tau}) \circ \eta \circ \varepsilon^{-1} \circ \epsA_{A,\tau} \equiv \pastmod{0 \le \tau}
\quad
\epsA_{A,\tau_1} \circ \pastmod {\tau_1} (\epsA_{\futuremod{\tau_1} A, \tau_2}) \circ \pastmod {\tau_1} (\pastmod {\tau_2} (\delta))
  \equiv \epsA \circ \mu
\end{array}
\]

\begin{proposition}
It then follows that $\etaA_{A,0} \equiv \varepsilon^{-1}_{\pastmod{0} A} \circ \eta_A$ and $\epsA_{A,0} \equiv \varepsilon_A \circ \eta^{-1}_{\futuremod{0}A}$.
\end{proposition}

In $\Pshf$, $\etaA_{A,\tau}$ and $\epsA_{A,\tau}$ are given by identities
on $A$-values, respectively combined with $\tau \le t + \tau$ and monotonicity for $(t \dotminus \tau) + \tau \equiv t$.
For the latter, we crucially know $\tau \le t$ due to the side-condition included in the definition
of $\pastmod{-}$.

We note that modulo the time grades $\tau$, the above structure is 
analogous to the models of the Fitch-style presentation of S4~\cite{Clouston:FitchStyle}, 
where $\boxmodstd$ is modelled by an idempotent comonad, \ctxunlock\,by 
an idempotent monad, and boxing/unboxing by $\ctxunlock\! \dashv \boxmodstd$.
This is also why we present $\futuremod -$ and $\pastmod -$ as comonad- and monad-like.

\subsubsection{Computation Types}

For computation types, we 
require a \emph{$\futuremod{-}$-strong graded monad} $(\T,\etaT,\muT,\strT)$ on 
$\Cat$, with grades in 
$\mathbb{N}$.\footnote{As $\lambdaTau$ does not include
  sub-effecting (see \cref{sect:future-work}), a discretely graded monad $T$ suffices.}
In detail, this means a functor $\T : \mathbb{N} \to [ \Cat , \Cat ]$, 
together with natural transformations $\etaT_A : A \to \T\, 0\, A$ (the \emph{unit}), 
$\muT_{A,\tau_1,\tau_2} : \T\, \tau_1 (\T\, \tau_2\, A) \to \T\, (\tau_1 + \tau_2)\, A$ (the \emph{multiplication}), 
and $\strT_{A,B,\tau} : \futuremod \tau A \times \T\, B\, \tau \to \T\, (A \times B)\, \tau$ (the \emph{strength}), 
with the first two satisfying standard graded monad laws (see~\cite{Katsumata:GradedMonads} or $(\eta,\mu)$ of $\pastmod{-}$). 
Below we only present the laws for $\strT$ because it 
has a novel temporal aspect to it---its first argument appears under   
$\futuremod \tau$. As such, $\strT$ expresses that if we know an $A$-value will
be available after $\tau$ time units, we can push it into computations taking $\tau$-time to execute.

We say that $\T$ is a $\futuremod -$-strong graded monad following the 
parlance of Bierman and de Paiva~\cite{Bierman:ModalLogic}---in their work they model the 
possibility modality ${\raisebox{-0.2mm}{\scalebox{1.5}{$\diamond$}}} A$
as a $\boxmodstd$-strong monad. While the laws governing $\strT$ are not 
overly different from standard graded strength laws~\cite{Katsumata:GradedMonads}, 
we have to correctly account for $\futuremod -$ in them
\[
\begin{array}{c}
\strT_{A,B,0} \circ (\varepsilon^{-1}_A \times \etaT_A) \equiv \etaT_{A \times B}
\quad
\muT_{A \times B,\tau_1,\tau_2} \circ \T\, (\strT) \circ \strT \equiv \strT \circ (\delta^{-1} \times \muT)
\\[1.25ex]
\T\, (\mathsf{snd}) \circ \strT_{A,B,\tau} \equiv \mathsf{snd}
\quad
\T\, (\alpha) \circ \strT \circ (\m \times \id) \circ \alpha^{-1} \equiv \strT_{A,B \times C,\tau} \circ (\id \times \strT)
\end{array}
\]
where  
$\alpha_{A,B,C} : (A \times B) \times C \overset{\cong}{\to} A \times (B \times C)$, and 
$\m_{A,B,\tau}: \futuremod{\tau} A \times \futuremod{\tau} B \to \futuremod{\tau} (A \times B)$ \linebreak
witnesses that $\futuremod \tau$ is monoidal for $\times$, 
which follows from $\futuremod \tau$ being a right adjoint~\cite{MacLane:CatWM}.
Observe that it is the $\futuremod -$-strength that naturally gives $\T$ a temporal flavour---the 
rest of it is standard~\cite{Katsumata:GradedMonads}. 
Below we show that $\strT$ is also mathematically natural, admitting an analogous
characterisation to ordinary strength.

\begin{proposition}
Analogously to ordinary strong and enriched monads~\cite{Kock:StrongFunctors}, $\T$
having $\futuremod -$-strength is equivalent to \emph{$\futuremod -$-enrichment} of $\T$, 
given by morphisms 
${\futuremod \tau(A \expto B) \to (\T\,\tau\,A \expto \T\,\tau\,B)}$ 
respecting $\Cat$'s self-enrichment~\cite{Kelly:EnrichedCats} and $(\etaT,\muT)$.
\end{proposition}

In order to model operations $\op$ and $\tmkw{delay}$ in \cref{sect:interpretation-terms}, 
we require $\T$ to be equipped with \emph{algebraic operations}: 
we ask there to be families of natural transformations
$\opT_{A,\tau} : \gsem{A_\op} \times \futuremod{\tau_\op}(\gsem{B_\op} \expto \T\, \tau\, A) \to \T\, (\tau_\op + \tau)\, A$, 
for all $\op : \tysigop{A_\op}{B_\op}{\tau_\op} \in \Ops$, and 
$\delayT_{A,\tau'}\, \tau : \futuremod {\tau} (\T\, \tau'\, A) \to \T\, (\tau + \tau')\, A$, for all $\tau \in \mathbb{N}$, 
satisfying algebraicity laws~\cite{Plotkin:HandlingEffects}, which state that both
commute with $\muT$ and $\strT$, e.g., 
\[
\strT_{A, B, \tau + \tau'} \circ (\id \times \delayT\, \tau)
  \equiv 
    \delayT_{A \times B,\tau'}\, \tau \circ \futuremod{\tau}(\strT) \circ \m \circ (\delta_{A,\tau,\tau'} \times \id)
\]

In $\Pshf$, we can define $\T$ as the initial algebra
of a corresponding signature functor for operations $\op$ and $\tmkw{delay}$, 
analogously to the usual treatment of algebraic effects~\cite{Bauer:WhatIsAlgebraic}.
Concretely, such $\T$ is determined inductively by three cases
\[
\ruleinfer{}
  {a \in A(t)}
  {\mathsf{ret}\, a \in (\T\, 0\, A)(t)}
\hspace{1.3ex}
\ruleinfer{}
  {a \in \gsem{A_\op}(t) \\\\ k \in (\futuremod{\tau_\op}(\gsem{B_\op} \expto \T\, \tau\, A))(t)}
  {\mathsf{op}\, a\, k \in (\T\, (\tau_\op + \tau)\, A)(t)}
\hspace{1.3ex}
\ruleinfer{}
  {k \in \futuremod {\tau} (\T\, \tau'\, A)(t)}
  {\mathsf{delay}\, \tau\, k \in (\T\, (\tau + \tau')\, A)(t)}
\]
with $(\etaT,\muT,\strT,\opT,\delayT)$ defined in the expected way, e.g., 
$\strT$ is given by recursively traversing a computation of type $\T\, \tau\, B$
and moving the argument of type $\futuremod {\tau} A$ under $\mathsf{ret}$ cases, 
modifying $\tau$ when going under the $\mathsf{op}$ and $\mathsf{delay}$ cases. 

\subsection{Interpretation of Value and Computation Terms}
\label{sect:interpretation-terms}

The interpretation of values and computations is defined 
simultaneously. We only present the temporally interesting cases---full details are 
in the online appendix.

As $\lambdaTau$ does not have sub-effecting and includes 
enough type annotations for typing derivations to be unique, this 
interpretation is \emph{coherent} by construction.

\subsubsection{Values}

We assume a morphism
$\sem{\tmconst{f}} : \gsem{A_1} \times \ldots \times \gsem{A_n} \to \gsem{B}$ for 
every $\tmconst{f} : (A_1,\ldots,A_n) \to B$. We  
interpret a well-typed value $\Gamma \types V : X$ as a morphism
$\sem{\Gamma \types V : X} : \sem{\Gamma}\One  \to \sem{X}$ in $\Cat$ by induction 
on the given typing derivation. 

Most of the value cases are standard, and analogous to other calculi based on 
fine-grain call-by-value~\cite{Levy:FGCBV} and graded 
monads~\cite{Katsumata:GradedMonads}, using the Cartesian-closed structure 
of $\Cat$. The temporally interesting cases are \rulename{Var} and \rulename{Box}, given by
\[
\small
\begin{array}{l}
\sem{\Gamma,x \of X,\Gamma' \types x : X} 
~\defeq~ 
\sem{\Gamma,x \of X,\Gamma'}\One
\overset{\iota}{-\!\!\!\!\!\!\longrightarrow}
\sem{\Gamma'}\big(\sem{\Gamma}\One \times \sem{X}\big)
\\[0.5ex]
\hspace{4cm}
\overset{\e}{-\!\!\!\!\!\!\longrightarrow}
\pastmod{\ctxtime {\Gamma'}}\big(\sem{\Gamma}\One \times \sem{X}\big)
\overset{\epsP}{-\!\!\!\!\!\!\longrightarrow}
\sem{\Gamma}\One \times \sem{X}
\overset{\snd}{-\!\!\!\!\!\!\longrightarrow}
\sem{X}
\\[2ex]
\sem{\Gamma \types \tmbox[\tau]{V} : \tybox{\tau}{X}} 
~\defeq~
\sem{\Gamma}\One
\overset{\etaA}{-\!\!\!\!\!\!\longrightarrow}
\futuremod{\tau} \big(\pastmod{\tau}(\sem{\Gamma}\One)\big)
\overset{\futuremod{\tau}(\sem{V})}{-\!\!\!-\!\!\!-\!\!\!\longrightarrow}
\futuremod{\tau} \sem{X}
\end{array}
\]
where $\e_{A,\Gamma} : \sem{\Gamma}A \to \pastmod{\ctxtime \Gamma}A$
extracts and collapses all temporal modalities in $\Gamma$, and
the counit-like $\epsP_{A,\tau}$ is given by the composite 
$\pastmod{\tau} A \overset{\pastmod{0 \le \tau}_A}{-\!\!\!-\!\!\!-\!\!\!\to} 
  \pastmod{0} A \overset{\eta^{-1}_A}{-\!\!\!\to} A$.
  
\subsubsection{Computations}

We interpret a well-typed computation 
$\Gamma \types M : \tycomp{X}{\tau}$ as a morphism
$\sem{\Gamma \types M : \tycomp{X}{\tau}} : \sem{\Gamma}{\One} \to \T\,\tau\,\sem{X}$ 
in $\Cat$ by induction on the typing derivation. The definition is 
largely unsurprising and follows a pattern similar to~\cite{Katsumata:GradedMonads,Levy:FGCBV}---the novelty 
lies in controlling the occurrences of $\pastmod -$ and $\futuremod -$.

In \rulename{Let}, we use $\pastmod{\tau} \dashv \futuremod{\tau}$
to push the environment ``into the future'', and then follow the standard 
monadic strength-followed-by-multiplication pattern~\cite{Katsumata:GradedMonads,Moggi:ComputationalLambdaCalculus}:
\[
\small
\begin{array}{l}
\sem{\Gamma \types \tmlet{x}{M}{N} : \tycomp{Y}{\tau + \tau'}}
~\defeq~
\sem{\Gamma} \One
\overset{\langle \etaA , \sem{M} \rangle}{-\!\!\!-\!\!\!-\!\!\!-\!\!\!\longrightarrow}
\futuremod{\tau}\big(\pastmod{\tau}(\sem{\Gamma} \One)\big) \times \T\,\tau\,\sem{X}
\\[1ex]
\hspace{2cm}
\overset{\strT}{-\!\!\!\!\!\!\longrightarrow}
\T\, \tau\, \big(\pastmod{\tau}(\sem{\Gamma} \One) \times \sem{X}\big)
\overset{\T\,(\sem{N})}{-\!\!\!-\!\!\!-\!\!\!-\!\!\!\longrightarrow}
\T\, \tau\, (\T\, \tau'\, \sem{Y})
\overset{\muT}{-\!\!\!\!\!\!\longrightarrow}
\T\, (\tau + \tau')\, \sem{Y}
\end{array}
\]
An analogous use of $\pastmod{\tau} \dashv \futuremod{\tau}$ also appears
in the cases for operations, e.g., in
\[
\small
\begin{array}{l}
\sem{\Gamma \types \tmop{op}{V}{\tmcont x M} : \tycomp{X}{\tau_\op + \tau}}
~\defeq~
\sem{\Gamma} \One
\overset{\langle \sem{V} , \etaA \rangle}{-\!\!\!-\!\!\!-\!\!\!-\!\!\!-\!\!\!\longrightarrow}
\gsem{A_\op} \times \futuremod{\tau_\op}\big(\pastmod{\tau_\op}(\sem{\Gamma} \One)\big)
\\[1ex]
\hfill
\overset{\!\!\id \times \futuremod{\tau_\op}(\mathsf{curry}(\sem{M}))}{-\!\!\!-\!\!\!-\!\!\!-\!\!\!-\!\!\!-\!\!\!-\!\!\!-\!\!\!-\!\!\!-\!\!\!-\!\!\!-\!\!\!\longrightarrow\,\,}
\gsem{A_\op} \times \futuremod{\tau_\op}\big(\gsem{B_\op} \expto \T\, \tau\, \sem{X}\big)
\overset{\opT}{-\!\!\!\!\!\!\longrightarrow}
\T\, (\tau_\op + \tau)\, \sem{X}
\end{array}
\]

Next, the \rulename{Unbox} case of the interpretation is defined as
\[
\small
\begin{array}{l}
\sem{\Gamma \types \tmunbox[\tau] V x N : \tycomp{Y}{\tau'}}
~\defeq~
\sem{\Gamma}\One 
\overset{\langle \id , \etaPRA \rangle}{-\!\!\!-\!\!\!-\!\!\!-\!\!\!\longrightarrow}
\sem{\Gamma}\One \times \pastmod{\tau}\big(\sem{\Gamma \ctxminus \tau} \One\big)
\\[1ex]
\hspace{2.5cm}
\overset{\!\!\id \times \pastmod{\tau}(\sem{V})}{-\!\!\!-\!\!\!-\!\!\!-\!\!\!-\!\!\!-\!\!\!\longrightarrow\,\,}
\sem{\Gamma}\One \times \pastmod{\tau}\big(\futuremod{\tau} \sem{X}\big)
\overset{\id \times \epsA}{-\!\!\!-\!\!\!\longrightarrow}
\sem{\Gamma}\One \times \sem{X}
\overset{\sem{N}}{-\!\!\!\!\!\!\longrightarrow}
\T\,\tau'\, \sem{Y}
\end{array}
\]
showing that temporal resources follow the common pattern
in which elimination forms are modelled by counits of adjunctions, whereas units 
model introduction forms (akin to functions). The morphism 
$\etaPRA_{\Gamma,A,\tau} : \sem{\Gamma}A \to \pastmod{\tau}(\sem{\Gamma \ctxminus \tau} A)$
extracts and collapses $\tau$ worth of context modalities in $\Gamma$, as long 
as $\tau \le \ctxtime \Gamma$. 
It is a semantic counterpart to an observation that the context modality $\Gamma, \ctxmod \tau$
is a \emph{parametric right adjoint} to the $\Gamma \ctxminus \tau$ operation, as in recent dependently
typed presentations of Fitch-style modal types~\cite{Gratzer:PRA}, see \cref{sect:related-work}
for further discussion.

Finally, we discuss the interpretation of effect handling. For
this, we additionally require $\Cat$ to have \emph{set-indexed products} 
$\Pi_{i \in I} A_i$ and \emph{handling morphisms}
\[
\begin{array}{l}
\h_{A,\tau,\tau'} : \Pi_{\op \in \Ops} \Pi_{\tau'' \in \mathbb{N}} \big( 
  (\gsem{A_\op} \times \futuremod{\tau_\op}(\gsem{B_\op} \expto \T\, \tau''\, A)) \expto \T\, (\tau_\op + \tau'')\, A \big)
\\[1ex]
\hspace{7.3cm}
\to \T\, \tau\, (\T\, \tau'\, A) \expto \T\, (\tau + \tau')\, A
\end{array}
\] 
satisfying laws which state that $\h_A$ returns a graded 
$\T$-algebra~\cite{Fuji:GradedMonads,McDermott:GradedAlgebras}, e.g., 
we require $\mathsf{uncurry}(\h_{A,0,\tau'}) \circ (\id \times \etaT) \equiv \snd$, 
where $\mathsf{uncurry}$ (and $\mathsf{curry}$ earlier) is part of the universal property 
of $A \expto B$. We also require similar laws for $\h$'s interaction with $\opT$ and $\delayT$.
In $\Pshf$, $\h$ is defined by recursively traversing a given tree, replacing 
all occurrences of $\mathsf{op}\, a\, k$ with respective operation clauses.

Writing $\mathcal{H}$ for the domain of $\h_{\sem{Y},\tau,\tau'}$, the \rulename{Handle} case is then defined as
\[
\small
\begin{array}{l}
\sem{\Gamma \types \tmhandle{M}{H}{x}{N} : \tycomp{Y}{\tau + \tau'}}
~\defeq~
\\[1ex]
\hspace{0.5cm}
\sem{\Gamma} \One
\overset{\langle \id , \langle \etaA , \sem{M} \rangle \rangle}{-\!\!\!-\!\!\!-\!\!\!-\!\!\!-\!\!\!-\!\!\!-\!\!\!-\!\!\!\longrightarrow}
\sem{\Gamma} \One \times \Big(\futuremod{\tau}\big(\pastmod{\tau}(\sem{\Gamma} \One)\big) \times \T\, \tau\, \sem{X}\Big)
\\[1ex]
\hspace{1cm}
\overset{\!\!\id \times \strT}{-\!\!\!-\!\!\!\longrightarrow\,\,}
\sem{\Gamma} \One \times \T\, \tau\, \big(\pastmod{\tau}(\sem{\Gamma} \One) \times \sem{X} \big)
\overset{\id \times \T\, \tau\, (\sem{N})}{-\!\!\!-\!\!\!-\!\!\!-\!\!\!-\!\!\!-\!\!\!-\!\!\!\longrightarrow}
\sem{\Gamma} \One \times \T\, \tau\, \big(\T\, \tau'\, \sem{Y}\big)
\\[1ex]
\hspace{4cm}
\overset{\sem{H} \times \id}{-\!\!\!-\!\!\!-\!\!\!\longrightarrow}
\mathcal{H} \times \T\, \tau\, \big(\T\, \tau'\, \sem{Y}\big)
\overset{\mathsf{uncurry}(\h_{\!\sem{Y},\tau,\tau'})}{-\!\!\!-\!\!\!-\!\!\!-\!\!\!-\!\!\!-\!\!\!-\!\!\!-\!\!\!-\!\!\!-\!\!\!\longrightarrow}
\T\, (\tau + \tau')\, \sem{Y}
\end{array}
\]
where we write $\sem{H}$ for the point-wise interpretation of operation clauses
\[
\sem{\Gamma}\One
\overset{\langle \langle \id \rangle_{\tau'' \in \mathbb{N}} \rangle_{\op \in \Ops}}{-\!\!\!-\!\!\!-\!\!\!-\!\!\!-\!\!\!-\!\!\!-\!\!\!-\!\!\!\longrightarrow}
\Pi_{\op \in \Ops} \Pi_{\tau'' \in \mathbb{N}} \Big(\sem{\Gamma}\One\Big)
\overset{\Pi_{\op \in \Ops} \Pi_{\tau'' \in \mathbb{N}} \big( \mathsf{curry}(\sem{M_\op\, \tau''} \,\circ\, \alpha^{-1}) \big)}{-\!\!\!-\!\!\!-\!\!\!-\!\!\!-\!\!\!-\!\!\!-\!\!\!-\!\!\!-\!\!\!-\!\!\!-\!\!\!-\!\!\!-\!\!\!-\!\!\!-\!\!\!-\!\!\!-\!\!\!-\!\!\!-\!\!\!-\!\!\!-\!\!\!-\!\!\!-\!\!\!\longrightarrow}
\mathcal{H}
\]

\subsection{Renamings, Substitutions, and Soundness}

We now show how syntactic renamings and substitutions 
relate to semantic morphism composition, using which 
we then prove the interpretation to be \emph{sound}.

\begin{proposition}
\label{prop:renaming-semantic}
Given $\rho : \Gamma \,{\leadsto}\, \Gamma'$ and $\Gamma \types J$, then 
$\sem{J[\rho]} \equiv \sem{J} \circ \sem{\rho}_{\One}$, where the 
interpretation of renamings $\sem{\rho}_A : \sem{\Gamma'}A \to \sem{\Gamma}A$ 
is defined by induction on the derivation of 
$\rho : \Gamma \leadsto \Gamma'$, with the morphism $\sem{\rho}_A$ also natural in $A$.
\end{proposition}

\begin{proposition}
\label{prop:substitution-semantic}
Given $\Gamma, x \of X, \Gamma' \types J$ and $\Gamma \types W : X$, we have 
$\sem{J[W/x]} \equiv \sem{J} \circ \iota_{\Gamma,x \of X;\Gamma';\One}^{-1} \circ \sem{\Gamma'}\big(\langle \id , \sem{W} \rangle\big) \circ \iota_{\Gamma;\Gamma';\One}$, where $(\iota,\iota^{-1})$ are discussed in \cref{sect:interpretation-types}.
\end{proposition}

\begin{proof}
We prove both results by induction on the derivation of ${\Gamma \types J}$. 
The proofs are unsurprising but require us to prove 
auxiliary lemmas about recursively defined renamings and 
semantic morphisms.
For example, for \cref{prop:renaming-semantic}, 
we show  
$\etaPRA \circ \sem{\rho} \equiv \pastmod{\tau}(\sem{\rho \ctxminus \tau}) \circ \etaPRA 
  : \sem{\Gamma'} A \to \pastmod{\tau}(\sem{\Gamma \ctxminus \tau}A)$, and
for \cref{prop:substitution-semantic}, that  
$\etaPRA \circ \iota \equiv \pastmod{\tau}\big(\iota\big) \circ \etaPRA 
  : \sem{\Gamma, \Gamma'} A \to \pastmod{\tau}\big(\sem{\Gamma' \ctxminus \tau}(\sem{\Gamma} A)\big)$, 
when $\tau \le \ctxtime \Gamma'$.
\end{proof}

\begin{theorem}
Given $\Gamma \types I \equiv J$ derived using the rules in \cref{sect:equational-theory}, 
then $\sem{I} \equiv \sem{J}$.
\end{theorem}

\begin{proof}
The proof proceeds by induction on the derivation of $\Gamma \types I \equiv J$, 
using \cref{prop:renaming-semantic} and \cref{prop:substitution-semantic} to unfold
the renamings and substitutions in the equations of \cref{sect:equational-theory}, 
and using the properties of the abstract structure we required $\Cat$ to have.
\end{proof}

% !TEX root = temporal-resources.tex

\section{Quotienting Delays}
\label{sect:delay-equations}

Observe that in $\lambdaTau$ the computations 
$\tmdelay \tau {(\tmdelay {\tau'} M)}$ and $\tmdelay {(\tau + \tau')} M$ cannot be proved 
equivalent, though in some situations this might be desired.

In order to deem the above two programs (and others alike) equivalent, we extend $\lambdaTau$'s 
equational theory with the following natural equations for $\tmkw{delay}$s:
\[
\begin{array}{c}
\tmdelay 0 M \equiv M
\qquad
\tmdelay \tau {(\tmdelay {\tau'} M)} \equiv \tmdelay {(\tau + \tau')} M
\end{array}
\]

\begin{theorem}
If the algebraic operations $\delayT\!$ of $\T$ satisfy analogous two
equations, the interpretation of \cref{sect:semantics} is sound for this extended equational theory.
\end{theorem}

For the concrete model in $\Pshf$, we have to \emph{quotient} $\T$~\cite{Katsumata:FlexiblePresentations}
by these two equations---the resulting graded monad is determined inductively by the cases
\[
%\small
\ruleinfer{}
  {k \in (\TS\, \tau\, A)(t)}
  {\mathsf{comp}\, k \in (\T\, \tau\, A)(t)}
\quad
\ruleinfer{}
  {\tau > 0 \\ k \in \futuremod {\tau} (\TS\, \tau'\, A)(t)}
  {\mathsf{delay}\, \tau\, k \in (\T\, (\tau + \tau')\, A)(t)}
\]
\[
%\small
\ruleinfer{}
  {a \in A(t)}
  {\mathsf{ret}\, a \in (\TS\, 0\, A)(t)}
\quad
\ruleinfer{}
  {a \in \gsem{A_\op}(t) \\ k \in (\futuremod{\tau_\op}(\gsem{B_\op} \expto \T\, \tau\, A))(t)}
  {\mathsf{op}\, a\, k \in (\TS\, (\tau_\op + \tau)\, A)(t)}
\vspace{0.1cm}
\]
where $(\T\, \tau\, A)(t)$ and $(\TS\, \tau\, A)(t)$ are defined simultaneously in such a way that 
only non-zero, non-consecutive $\mathsf{delay}$s can appear in the tree structure. 

% !TEX root = temporal-resources.tex

\section{Related and Future Work}
\label{sect:related-future-work}

\subsection{Related Work}
\label{sect:related-work}

We contribute to two prominent areas: (i) modal types and (ii) graded monads.

As noted in \cref{sect:overview-modal-types}, \emph{modal types} provide a mathematically natural
means for capturing many aspects of programming. Adding to~\cref{sect:overview-modal-types}, types 
corresponding to the \emph{eventually} and \emph{always modalities} of temporal logics capture \emph{functional 
reactive programming (FRP)}~\cite{Cave:FairFRP,Jeltsch:FRP,Krishnaswami:FRP}, including a 
combination with linearity and time-annotations to model resources~\cite{Jeltsch:Resources}, 
where \emph{all} values are annotated with inhabitation times.
Recently, FRP has also been studied in Fitch-style~\cite{Bahr:SimplyRaTT}.
Starting with Nakano~\cite{Nakano:Guarded}, modal types 
have also been used for \emph{guarded recursion}, even in the 
dependently typed setting~\cite{Bahr:ClocksTicking,Bizjak:GuardedTT,Mannaa:ClockSemantics}, 
including in Fitch-style~\cite{Birkedal:DepRightAdjoints}.

We also note that $\lambdaTau$'s time grades $\tau$ and the $\Gamma \ctxminus \tau$ operation are 
closely related to recent dependently typed Fitch-style frameworks. Namely,~\cite{Gratzer:MTT} 
develops a \emph{multimodal type theory} (MTT) where types $\futuremod \mu X$ are 
indexed by 1-cells $\mu$ of a strict 2-category (a mode theory). 
The time grades $\tau$ of $\lambdaTau$ are an example of such mode theories, given by the 
delooping of $\mathbb{N}$, i.e., by a single 0-cell, $\tau$s as 1-cells, and $\tau \le \tau'$s as 2-cells. While 
ensuring the admissibility of and naturality under substitutions, MTT with its indirect elimination rule 
for $\futuremod \mu X$ is weaker than earlier systems (such as~\cite{Birkedal:DepRightAdjoints}). The direct-style elimination rule is recovered 
in~\cite{Gratzer:PRA} by observing that in addition to $\Gamma, \ctxmod \mu$ being a left adjoint 
to $\futuremod \mu X$, it should further form a 
\emph{parametric right adjoint (PRA)}~\cite{Carboni:ConnLimits,Weber:FamilialFunctors}
to contexts of the form $\Gamma / (r : \mu)$, where $r$ is a substitution  
$\cdot, \ctxmod \mu \leadsto \Gamma$. The operation $\Gamma \ctxminus \tau$ in $\lambdaTau$
is an instance of this: $\mu$ is a $\tau$, $r$ corresponds to the condition $\tau \le \ctxtime \Gamma$ 
in \rulename{Unbox}, contexts $\Gamma / (r : \mu)$ are given by $\Gamma \ctxminus \tau$, 
and the PRA situation is witnessed by renamings 
$((\Gamma \ctxminus \tau), \ctxmod \tau) \leadsto \Gamma$, when $\tau \le \ctxtime \Gamma$, and 
$\Gamma \leadsto ((\Gamma, \ctxmod \tau) \ctxminus \tau)$.

\emph{Graded monads} provide a uniform framework for  
different effect systems and effect-based  
analyses~\cite{Fuji:GradedMonads,Katsumata:GradedMonads,Katsumata:FlexiblePresentations,McDermott:GradedAlgebras,Mellies:GradedMonads}. 
A major contribution of ours is showing that 
context modalities can inform continuations of preceding computations' effects.
While the theory of graded 
monads
can be instantiated with any ordered monoid, we focus on natural
numbers to model time, but do not expect 
complications generalising $\lambdaTau$ to other structures with 
same properties as $(\mathbb{N}, 0, +, \dotminus, \le)$, 
and perhaps even to grading $T$ and $\pastmod -$, $\futuremod -$ with different structures, akin to~\cite{Gaboardi:EffectsCoefects}.

Our use of $\tybox{\tau}{X}$ to restrict when resources
are available is somewhat reminiscent of 
\emph{coeffects}~\cite{Brunel:Coeffects,Ghica:ResourceSemiring,Petricek:Coeffects,Petricek:CoeffectCalculus} 
and \emph{quantitative type systems}~\cite{Atkey:QTT,McBride:QTT,Moon:GrQTT}. In these works, 
variables are graded by (semi)ring-valued $r$s, as $x \of_r X$, counting how many times and 
in which ways $x$ is used, enabling applications such as liveness and dataflow 
analyses~\cite{Petricek:Coeffects}. Semantically, these systems often interpret $x \of_r X$
using a graded comonad, as $\boxmodstd_r X$, where one can access $X$ only if $r \equiv 1$. 
Of such works, the closest to ours is that of Gaboardi et al.~\cite{Gaboardi:EffectsCoefects}, 
who combine coeffects with effectful programs via distributive laws between the grades of 
coeffects and effects, allowing coeffectful analyses to be propagated through effectful computations. 

We also note that the type $\tybox{\tau}{X}$ can be intuitively also viewed as a 
temporally-graded variant of \emph{promise types}~\cite{Haller:Futures,Schwinghammer:Thesis}, 
in that it expresses that a value of type $X$ will be available in the future, but with 
additional time guarantees.

\subsection{Future Work}
\label{sect:future-work}

Currently, $\lambdaTau$ does not support
\emph{sub-effecting}: we cannot deduce from ${\tau \le \tau'}$ and ${\Gamma \types M : \tycomp{X}{\tau}}$ 
that ${\Gamma \types M : \tycomp{X}{\tau'}}$. Of course, we can simulate 
this by inserting $\tau' \!\dotminus \tau$ worth of explicit $\tmkw{delay}$s into 
$M$, but this is extremely intensional, fixing where $\tmkw{delay}$s happen. 
In particular, we cannot type equations such as $\tmlet{x}{(\tmreturn{V})}{N} \equiv N[V/x]$
if $\tmreturn{V}$ was sub-effected to $\tau > 0$, with the 
$\ctxmod \tau$ in $N$'s context the culprit. However, when considering sub-effecting 
as a \emph{coercion} $\mathsf{coerce}_{\tau \le \tau'}\, M$, we believe we can 
add it by considering equations stating that it will
produce \emph{all the possible ways} how $\tau' \!\dotminus \tau$ worth of $\tmkw{delay}$s 
could be inserted into $M$. Of course, this will require a more complex non-deterministic semantics.
 
It would be neat if $\lambdaTau$ also included \emph{recursion} in a way that 
programs could make use of the temporal discipline.
This is likely unattainable for general recursion, but we hope that \emph{primitive recursion}
(say, on natural numbers) can be added via \emph{type-dependency}
of time grades $\tau$ on the values being recursed on.

It would be interesting to combine $\lambdaTau$ with linear~\cite{Girard:LinearLogic} 
and separation logics~\cite{Jung:Iris,Reynolds:SeparationLogic} to  
model \emph{linear} and \emph{spatial properties} of temporal resources. 
Another goal would be to add \emph{concurrency}, e.g., using 
(multi)handlers~\cite{Bauer:AlgebraicEffects,Convent:DooBeeDooBeeDoo,Dolan:MulticoreOCaml}.
We also plan to look into capturing \emph{expiring} and 
\emph{available-for-an-interval} style resources.

Further, we plan to study $\lambdaTau$'s \emph{operational semantics}, 
namely, one that takes time seriously and does not model $\tmkw{delay}$s simply 
as uninterpreted operations~\cite{Bauer:AlgebraicEffects}, 
together with developing a \emph{prototype}, and proving 
\emph{normalisation} akin to~\cite{Gratzer:MTTNBE,Valliappan:NBE}. 

We also plan to study the \emph{completeness} of the denotational semantics of $\lambdaTau$.
For such semantic investigations, it could be beneficial to also study the 
general theory of the kinds of temporally aware graded algebraic effects used in this paper, 
by investigating their \emph{algebras} and \emph{equational 
presentations}~\cite{Katsumata:FlexiblePresentations,McDermott:GradedAlgebras}.

% !TEX root = temporal-resources.tex

\section{Conclusion}
\label{sect:conclusion}

We have shown how a temporal, time-graded variant of Fitch-style modal type systems, 
when combined with an effect system based on graded monads, provides a natural 
framework for safe programming with temporal resources.
To this end, we developed a modally typed, effectful, equationally-presented 
core calculus, and equipped it with a sound denotational semantics based on strong monoidal functors
(for modelling modalities) and graded monads (for modelling effects). The calculus also includes
temporally aware graded algebraic effects and effect handlers, with the continuations of the 
former knowing that an operation's worth of additional time has passed before they start executing, 
and where the user-defined effect handlers are guaranteed to respect this temporal discipline.

%%% Acknowledgements

\paragraph{Acknowledgements}

We thank Andrej Bauer, Juhan-Peep Ernits, Niccolò Veltri, and Niels Voorneveld for  
useful discussions. 
We also thank one of the reviewers for drawing our attention to the recent work on 
presenting Fitch-style modal types in terms of parametric right adjoints, and its relationship to the 
work presented in this paper.
This material is based upon work supported by the Air Force Office of 
Scientific Research under award number FA9550-21-1-0024.

%% BIBLIOGRAPHY
\bibliographystyle{splncs04}
\bibliography{references}

% APPENDIX
\newpage
\appendix
\makeatletter\def\@seccntformat#1{\appendixname~\csname the#1\endcsname: }\makeatother
% !TEX root = temporal-resources.tex

\section{Renamings}
\label{sect:appendix-renamings}

We expand here on the definitions and results that we use in the proof of \cref{thm:renaming}.

First, we present the full definition of the \emph{renaming relation} $\rho : \Gamma \leadsto \Gamma'$.
As noted in the paper, it is given by the reflexive-transitive-congruent closure of the desired structural rules. 
In detail, it is given by the following cases:

\parbox{\textwidth}{
  \mathtoolsset{original-shortintertext=false,below-shortintertext-sep=0pt,above-shortintertext-sep=0pt}
  \begin{align*}
    \mathsf{id}^r & : \Gamma \leadsto \Gamma
    \\
    \mathsf{\circ}^r & :  (\Gamma' \leadsto \Gamma'') \longrightarrow (\Gamma \leadsto \Gamma') \longrightarrow (\Gamma \leadsto \Gamma'')
    \\[1ex]
    \mathsf{wk}^r & : \Gamma \leadsto (\Gamma, x \of X)
    \\[1ex]
    \mathsf{var}^r_{x \of X \in \Gamma} & : (\Gamma, y \of X) \leadsto \Gamma
    \\[1ex]
    \eta^r & : (\Gamma, \ctxmod 0) \leadsto \Gamma
    \\
    (\eta^r)^{-1} & : \Gamma \leadsto (\Gamma, \ctxmod 0)
    \\[1ex]
    \mu^r & : (\Gamma, \ctxmod {\tau + \tau'}) \leadsto (\Gamma, \ctxmod \tau, \ctxmod {\tau'})
    \\
    (\mu^r)^{-1} & : (\Gamma, \ctxmod \tau, \ctxmod {\tau'}) \leadsto (\Gamma, \ctxmod {\tau + \tau'})
    \\[1ex]
    \mathsf{mon}^r_{\tau \le \tau'} & : (\Gamma, \ctxmod {\tau}) \leadsto (\Gamma, \ctxmod {\tau'})
    \\[1ex]
    \mathsf{cong\text{-}var}^r & : (\Gamma \leadsto \Gamma') \longrightarrow ((\Gamma, x \of X) \leadsto (\Gamma', x \of X)) 
    \\
    \mathsf{cong\text{-}mod}^r & : (\Gamma \leadsto \Gamma') \longrightarrow ((\Gamma, \ctxmod \tau) \leadsto (\Gamma', \ctxmod \tau)) 
    \\
  \end{align*}
}
As in the rest of the paper, we assume that all contexts involved in the definition are well-formed (i.e., 
consist of distinct variables), and all mentioned variables are fresh for the contexts that they extend 
(e.g., $y$ for $\Gamma$ in $\Gamma, y \of X$).

This renaming relation then has many desirable properties, some of which we use in the proof of
\cref{thm:renaming}. These results are all proved by either one-line arguments or by fairly straightforward 
induction on the derivations of the given judgements. Full details of these proofs can be found in the Agda formalisation.

\begin{proposition}[Basic Properties of the Renaming Relation]
\begin{enumerate}
\item If $\Gamma \leadsto \Gamma'$ and $x \in \Gamma'$, then we have $(\Gamma, y \of X) \leadsto \Gamma'$.
\item If $\Gamma \leadsto \Gamma'$, then we have $(\Gamma, \Gamma'') \leadsto (\Gamma', \Gamma'')$.
\item $\Gamma \leadsto (\Gamma,\Gamma')$
\item $\Gamma \leadsto (\Gamma, \ctxmod \tau)$
\item $(\Gamma, x \of X, y \of Y) \leadsto (\Gamma, y \of Y, x \of X)$
\item $(\Gamma, \ctxmod \tau, x \of X) \leadsto (\Gamma, x \of X, \ctxmod \tau)$
\item $(\Gamma, x \of X, y \of X) \leadsto (\Gamma, x \of X)$
\item If $\Gamma \leadsto \Gamma'$, then we have $\ctxtime \Gamma \le \ctxtime {\Gamma'}$.
\item If $\tau \le \ctxtime \Gamma$, then we have $((\Gamma \ctxminus \tau), \ctxmod \tau) \leadsto \Gamma$.
\item $\Gamma \leadsto ((\Gamma, \ctxmod \tau) \ctxminus \tau)$
\item $(\Gamma \ctxminus \tau) \leadsto \Gamma$
\item If $\tau_1 \le \tau_2$, then we have $(\Gamma \ctxminus \tau_2) \leadsto (\Gamma \ctxminus \tau_1)$.
\item If $\rho : \Gamma \leadsto \Gamma'$ and $x \of X \in_\tau \Gamma$, 
  then $\rho\, x \of X \in_{\tau'} \Gamma'$ for some $\tau'$ with $\tau \le \tau'$.
\item If $\rho : \Gamma \leadsto \Gamma'$, then we have $\rho \ctxminus \tau : (\Gamma \ctxminus \tau) \leadsto (\Gamma' \ctxminus \tau)$ .
\end{enumerate}
\end{proposition}

In (13), the relation $x \of X \in_\tau \Gamma$, is defined inductively by the following cases:
\[
\ruleinfer{}
  {\phantom{...}}
  {x \of X \in_0 \Gamma, x \of X}
\qquad
\ruleinfer{}
  {x \of X \in_\tau \Gamma \quad x \neq y}
  {x \of X \in_\tau \Gamma, y \of Y}
\qquad
\ruleinfer{}
  {x \of X \in_\tau \Gamma }
  {x \of X \in_{\tau + \tau'} \Gamma, \ctxmod {\tau'}}
\]
Intuitively, this relation captures that if $x \of X \in_\tau \Gamma$, then $x \of X \in \Gamma$ and 
there is $\tau$ worth of context modalities to the right of $x$ in $\Gamma$, i.e., $x$ was brought into 
scope $\tau$ time units ago.

In (14), the operation $\rho \ctxminus \tau : (\Gamma \ctxminus \tau) \leadsto (\Gamma' \ctxminus \tau)$
on renamings is defined by induction on the derivation of $\rho : \Gamma \leadsto \Gamma'$.
The definition is somewhat long but fairly straightforward---it proceeds by removing 
variables and context modalities from the right-hand sides of $\Gamma$ and $\Gamma'$ until 
$\tau$ becomes $0$, see the definition of the corresponding operation $\Gamma \ctxminus \tau$ 
on contexts for intuition. Properties such as (13) then ensure that variables in $\Gamma \ctxminus \tau$
are never mapped by $\rho$ to outside of $\Gamma' \ctxminus \tau$. Full details of this 
definition can be found in the Agda formalisation.

% !TEX root = temporal-resources.tex

\section{Denotational Semantics}
\label{sect:appendix-semantics}

In this appendix we spell out full details of the more important semantic definitions 
that were abbreviated or left out in \cref{sect:semantics} due to space constraints.

\subsection{Strong Monoidal Functor for Temporal Resources}

We require a functor
\[
\futuremod - : (\mathbb{N},\le) \to [ \Cat , \Cat ]
\]
and families of natural isomorphisms (counit and comultiplication)
\[
\varepsilon_A : \futuremod 0 A \overset{\cong}{\to} A
\qquad
\delta_{A,\tau_1,\tau_2} : \futuremod {\tau_1 + \tau_2} A \overset{\cong}{\to} \futuremod {\tau_1} (\futuremod {\tau_2} A)
\]
satisfying time-graded variants of the laws of a comonad
\[
\small
\xymatrix@C=3em@R=3.5em@M=0.5em{
\futuremod{0 + \tau} A \ar[r]^-{\delta_{A,0,\tau}} \ar[dr]_{\equiv} & \futuremod{0}(\futuremod{\tau} A) \ar[d]^{\varepsilon_{\futuremod{\tau}A}}
&
\futuremod{\tau + 0} A \ar[r]^-{\delta_{A,\tau,0}} \ar[dr]_{\equiv} & \futuremod{\tau}(\futuremod{0}A) \ar[d]^{\futuremod{\tau}(\varepsilon_A)}
\\
& \futuremod{\tau} A
&
& \futuremod{\tau} A
}
\]
\[
\small
\xymatrix@C=5em@R=3.5em@M=0.5em{
\futuremod{(\tau_1 + \tau_2) + \tau_3} A \ar[d]_{\equiv} \ar[r]^{\delta_{A,\tau_1+\tau_2,\tau_3}} & \futuremod{\tau_1 + \tau_2}(\futuremod{\tau_3} A) \ar[r]^{\delta_{\futuremod{\tau_3}A, \tau_1, \tau_2}} & \futuremod{\tau_1}(\futuremod{\tau_2}(\futuremod{\tau_3} A)) \ar[d]^{\id}
\\
\futuremod{\tau_1 + (\tau_2 + \tau_3)} A \ar[r]_{\delta_{A,\tau_1,\tau_2 + \tau_3}} & \futuremod{\tau_1}(\futuremod{\tau_2 + \tau_3} A) \ar[r]_{\futuremod{\tau_1}(\delta_{A,\tau_2,\tau_3})} & \futuremod{\tau_1}(\futuremod{\tau_2}(\futuremod{\tau_3} A))
}
\]
and $(\delta_{A,\tau_1,\tau_2},\delta^{-1}_{A,\tau_1,\tau_2})$ have to be monotone in 
$\tau_1, \tau_2$: if $\tau_1 \le \tau'_1$ and $\tau_2 \le \tau'_2$, then
\[
\small
\xymatrix@C=5em@R=3.5em@M=0.5em{
\futuremod{\tau_1 + \tau_2} A \ar[rr]^{\futuremod{\tau_1 + \tau_2 \le \tau'_1 + \tau'_2}_A} \ar[d]_{\delta_{A,\tau_1,\tau_2}} & & \futuremod{\tau'_1 + \tau'_2} A \ar[d]^{\delta_{A,\tau'_1,\tau'_2}}
\\
\futuremod{\tau_1}(\futuremod{\tau_2} A) \ar[r]_{\futuremod{\tau_1 \le \tau'_1}_{\futuremod{\tau_2} A}} & \futuremod{\tau'_1}(\futuremod{\tau_2} A) \ar[r]_{\futuremod{\tau'_1}(\futuremod{\tau_2 \le \tau'_2}_A)} & \futuremod{\tau'_1}(\futuremod{\tau'_2} A)
}
\vspace{0.2cm}
\]
A similar diagram follows for $\delta^{-1}_{A}$ from 
$(\delta_{A,\tau_1,\tau_2},\delta^{-1}_{A,\tau_1,\tau_2})$ forming isomorphisms.

\subsection{Strong Monoidal Functor for Context Modalities}

We require a contravariant functor 
\[
\pastmod - : (\mathbb{N},\le)^{\text{op}} \to [ \Cat , \Cat ]
\]
and families of natural isomorphisms (unit and multiplication)
\[
\eta_A : A \overset{\cong}{\to} \pastmod 0 A 
\qquad
\mu_{A, \tau_1, \tau_2} : \pastmod {\tau_1} (\pastmod {\tau_2} A) \overset{\cong}{\to} \pastmod {\tau_1 + \tau_2} A
\]
satisfying time-graded variants of the laws of a monad
\[
\small
\xymatrix@C=3em@R=3.5em@M=0.5em{
\pastmod{\tau} A \ar[r]^-{\eta_{\pastmod{\tau} A}} \ar[dr]_{\equiv} & \pastmod{0}(\pastmod{\tau}A) \ar[d]^{\mu_{A,0,\tau}}
&
\pastmod{\tau}A \ar[r]^-{\pastmod{\tau}(\eta_A)} \ar[dr]_{\equiv} & \pastmod{\tau}(\pastmod{0} A) \ar[d]^{\mu_{A,\tau,0}}
\\
& \pastmod{0 + \tau}A
&
& \pastmod{\tau + 0}A
}
\]
\[
\small
\xymatrix@C=5em@R=3.5em@M=0.5em{
\pastmod{\tau_1}(\pastmod{\tau_2}(\pastmod{\tau_3} A)) \ar[r]^{\mu_{\pastmod{\tau_3} A, \tau_1, \tau_2}} \ar[d]_{\id} & \pastmod{\tau_1 + \tau_2}(\pastmod{\tau_3} A) \ar[r]^{\mu_{A,\tau_1 + \tau_2,\tau_3}} & \pastmod{(\tau_1 + \tau_2) + \tau_3} A \ar[d]^{\equiv}
\\
\pastmod{\tau_1}(\pastmod{\tau_2}(\pastmod{\tau_3} A)) \ar[r]_{\pastmod{\tau_1}(\mu_{A,\tau_2,\tau_3})} & \pastmod{\tau_1}(\pastmod{\tau_2 + \tau_3} A) \ar[r]_{\mu_{A,\tau_1,\tau_2 + \tau_3}} & \pastmod{\tau_1 + (\tau_2 + \tau_3)} A
}
\vspace{0.2cm}
\]
and $(\mu_{A,\tau_1,\tau_2},\mu^{-1}_{A,\tau_1,\tau_2})$ have to be monotone in 
$\tau_1, \tau_2$: if $\tau_1 \!\le\! \tau'_1$ and $\tau_2 \!\le\! \tau'_2$, then
\[
\small
\xymatrix@C=5em@R=3.5em@M=0.5em{
\pastmod{\tau'_1}(\pastmod{\tau'_2} A) \ar[d]_{\mu_{A,\tau'_1,\tau'_2}} \ar[r]^{\pastmod{\tau_1 \le \tau'_1}_{\pastmod{\tau'_2} A}} & \pastmod{\tau_1}(\pastmod{\tau'_2} A) \ar[r]^{\pastmod{\tau_1}(\pastmod{\tau_2 \le \tau'_2}_A)} & \pastmod{\tau_1}(\pastmod{\tau_2} A) \ar[d]^{\mu_{A,\tau_1,\tau_2}}
\\
\pastmod{\tau'_1 + \tau'_2} A \ar[rr]_{\pastmod{\tau_1 + \tau_2 \le \tau'_1 + \tau'_2}_A} && \pastmod{\tau_1 + \tau_2} A
}
\]
A similar diagram follows for $\mu^{-1}_{A}$ from 
$(\mu_{A,\tau_1,\tau_2},\mu^{-1}_{A,\tau_1,\tau_2})$ forming isomorphisms.

\subsection{Adjunctions for Boxing and Unboxing Resources}

We require adjunctions
\[
\pastmod \tau \dashv \futuremod \tau
\]
for all $\tau \in \mathbb{N}$, i.e., families of natural transformations (unit and counit)
\[
\etaA_{A,\tau} : A \to \futuremod{\tau} (\pastmod{\tau} A)
\qquad
\epsA_{A,\tau} : \pastmod{\tau} (\futuremod{\tau} A) \to A
\]
satisfying the two standard adjunction laws
\[
\small
\xymatrix@C=4em@R=3.5em@M=0.5em{
\pastmod{\tau}A \ar[r]^-{\pastmod{\tau}(\etaA_{A,\tau})} \ar[dr]_{\id} & \pastmod{\tau}(\futuremod{\tau}(\pastmod{\tau}A)) \ar[d]^{\epsA_{\pastmod{\tau}A,\tau}}
&
\futuremod{\tau}A \ar[r]^-{\etaA_{\futuremod{\tau}A,\tau}} \ar[dr]_{\id} & \futuremod{\tau}(\pastmod{\tau}(\futuremod{\tau} A)) \ar[d]^{\futuremod{\tau}(\epsA_{A,\tau})}
\\
& \pastmod{\tau}A
&
& \futuremod{\tau}A
}
\]
and interacting well with the strong-monoidal structure of $\pastmod -$ and $\futuremod -$
\[
\small
\xymatrix@C=3em@R=3.5em@M=0.5em{
\futuremod{0}(\pastmod{0} A) \ar[rr]^-{\futuremod{0 \le \tau}_{\pastmod{0} A}} \ar[d]_{\varepsilon_{\pastmod{0} A}} & & \futuremod{\tau}(\pastmod{0} A)
\\
\pastmod{0}A \ar[r]_-{\eta^{-1}} & A \ar[r]_-{\etaA_{A,\tau}} & \futuremod{\tau}(\pastmod{\tau}A) \ar[u]_{\futuremod{\tau}(\pastmod{0 \le \tau}_A)}
}
\]
\[
\small
\xymatrix@C=3em@R=3.5em@M=0.5em{
\pastmod{\tau}(\futuremod{\tau} A) \ar[d]_{\epsA_{A,\tau}} \ar[rr]^{\pastmod{0 \le \tau}_{\futuremod{\tau} A}} && \pastmod{0}(\futuremod{\tau} A)
\\
A \ar[r]_-{\varepsilon^{-1}_A} & \futuremod{0} A \ar[r]_-{\eta_{\futuremod{0}A}} & \pastmod{0}(\futuremod{0} A) \ar[u]_{\pastmod{0}(\futuremod{0 \le \tau}_A)}
}
\]
\[
\small
\xymatrix@C=4em@R=3.5em@M=0.5em{
A \ar[r]^-{\etaA_{A,\tau_1}} \ar[d]_{\etaA_{A,\tau_1 + \tau_2}} & \futuremod{\tau_1}(\pastmod{\tau_2} A) \ar[r]^-{\futuremod{\tau_1}(\etaA_{\pastmod{\tau_1}A,\tau_2})} & \futuremod{\tau_1}(\futuremod{\tau_2}(\pastmod{\tau_2}(\pastmod{\tau_1} A))) \ar[d]^{\futuremod{\tau_1}(\futuremod{\tau_2}(\mu_{A,\tau_2,\tau_1}))}
\\
\futuremod{\tau_1 + \tau_2}(\pastmod{\tau_1 + \tau_2} A) \ar[rr]_-{\delta_{\pastmod{\tau_1 + \tau_2} A,\tau_1,\tau_2}} && \futuremod{\tau_1}(\futuremod{\tau_2}(\pastmod{\tau_1 + \tau_2} A))
}
\vspace{0.2cm}
\]
\[
\small
\xymatrix@C=4em@R=3.5em@M=0.5em{
\pastmod{\tau_1}(\pastmod{\tau_2}(\futuremod{\tau_1 + \tau_2} A)) \ar[r]^-{\pastmod{\tau_1}(\pastmod{\tau_2}(\delta_{A,\tau_2,\tau_1}))} \ar[d]_{\mu_{\futuremod{\tau_1 + \tau_2} A,\tau_1, \tau_2}} & \pastmod{\tau_1}(\pastmod{\tau_2}(\futuremod{\tau_2}(\futuremod{\tau_1} A))) \ar[r]^-{\pastmod{\tau_1}(\epsA_{\futuremod{\tau_1}A,\tau_2})} & \pastmod{\tau_1}(\futuremod{\tau_1} A) \ar[d]^{\epsA_{A,\tau_1}}
\\
\pastmod{\tau_1 + \tau_2}(\futuremod{\tau_1 + \tau_2}A) \ar[rr]_-{\epsA_{A,\tau_1 + \tau_2}} && A
}
\]

\subsection{$\futuremod -$-Strong Graded Monad for Computations}

We require a functor
\[
\T : \mathbb{N} \to [ \Cat , \Cat ]
\]
together with natural transformations (unit, multiplication, and strength)
\[
\etaT_A : A \to \T\, 0\, A
\qquad
\muT_{A,\tau_1,\tau_2} : \T\, \tau_1 (\T\, \tau_2\, A) \to \T\, (\tau_1 + \tau_2)\, A
\]
\[
\strT_{A,B,\tau} : \futuremod \tau A \times \T\, B\, \tau \to \T\, (A \times B)\, \tau
\]
satisfying standard graded monad laws
\[
\small
\xymatrix@C=3em@R=3.5em@M=0.5em{
\T\, \tau\, A \ar[dr]_{\equiv} \ar[r]^-{\etaT_{\T \tau A}} & \T\, 0\, (\T\, \tau\, A) \ar[d]^{\muT_{A,0,\tau}}
&
\T\, \tau\, A \ar[r]^-{\T\, \tau\, (\etaT_{A})} \ar[dr]_{\equiv} & \T\, \tau\, (\T\, 0\, A) \ar[d]^{\mu_{A,\tau,0}}
\\
& \T\, (0 + \tau)\, A
&
& \T\, (\tau + 0)\, A
}
\] 
\[
\small
\xymatrix@C=4.5em@R=3.5em@M=0.5em{
\T\, \tau_1\, (\T\, \tau_2\, (\T\, \tau_3\, A)) \ar[r]^-{\muT_{\T \tau_3 A,\tau_1,\tau_2}} \ar[d]_{\id} & \T\, (\tau_1 + \tau_2)\, (\T\, \tau_3\, A) \ar[r]^-{\muT_{A,\tau_1 + \tau_2,\tau_3}} & \T\, ((\tau_1 + \tau_2) + \tau_3)\, A \ar[d]^{\equiv}
\\
\T\, \tau_1\, (\T\, \tau_2\, (\T\, \tau_3\, A)) \ar[r]_-{\T\, \tau_1\, (\muT_{A,\tau_2,\tau_3})} & \T\, \tau_1\, (\T\, (\tau_2 + \tau_3)\, A) \ar[r]_-{\muT_{A,\tau_1,\tau_2 + \tau_3}} & \T\, (\tau_1 + (\tau_2 + \tau_3))\, A
}
\]
and $\futuremod -$-strength laws
\[
\small
\xymatrix@C=3em@R=3.5em@M=0.5em{
A \times B \ar[r]^-{\varepsilon^{-1}_A \times \etaT_A} \ar[dr]_{\etaT_{A \times B}} & \futuremod{0} A \times \T\, 0\, B \ar[d]^{\strT_{A,B,0}}
&
\futuremod{\tau}A \times \T\, \tau\, B \ar[r]^-{\strT_{A,B,\tau}} \ar[dr]_{\snd} & \T\, \tau\, (A \times B) \ar[d]^{\T\, \tau\, (\snd)}
\\
& \T\, 0\, (A \times B)
&
& \T\, \tau\, B
}
\]
\[
\hspace{-0.2cm}
\small
\xymatrix@C=2.5em@R=3.5em@M=0.5em{
\futuremod{\tau_1}(\futuremod{\tau_2}A) \times \T\, \tau_1\, (\T\, \tau_2\, B) \ar[r]^-{\strT_{\futuremod{\tau_2}A,\T \tau_2 B,\tau_1}} \ar[d]_{\delta^{-1}_{A,\tau_1,\tau_2} \times \muT_{B,\tau_1,\tau_2}} & T\, \tau_1\, (\futuremod{\tau_2} A \times T\, \tau_2\, B) \ar[r]^-{\T\, \tau_1\, (\strT_{A,B,\tau_2})} & \T\, \tau_1\, (\T\, \tau_2\, (A \times B)) \ar[d]^{\muT_{A \times B,\tau_1,\tau_2}}
\\
\futuremod{\tau_1 + \tau_2} A \times \T\, (\tau_1 + \tau_2)\, B \ar[rr]_-{\strT_{A,B,\tau_1 + \tau_2}} && \T\, (\tau_1 + \tau_2)\, (A \times B)
}
\]
\[
\small
\xymatrix@C=3em@R=3.5em@M=0.5em{
\futuremod{\tau} A \times (\futuremod{\tau} B \times \T\, \tau\, C) \ar[r]^-{\id \times \strT_{B,C,\tau}} \ar[d]_{\alpha^{-1}} & \futuremod{\tau} A \times \T\, \tau\, (B \times C) \ar[r]^-{\strT_{A,B \times C,\tau}} & \T\, \tau\, (A \times (B \times C))
\\
(\futuremod{\tau} A \times \futuremod{\tau} B) \times \T\, \tau\, C \ar[r]_-{\m_{A,B,\tau} \times \id} & \futuremod{\tau}(A \times B) \times \T\, \tau\, C \ar[r]_-{\strT_{A \times B,C,\tau}} & \T\, \tau\, ((A \times B) \times C) \ar[u]_{\T\, \tau\, \alpha}
}
\]

\subsection{$\futuremod -$-Enrichment of Graded Monads}

We require morphisms
\[
\enr_{A,B,\tau} : \futuremod \tau(A \expto B) \to (\T\,\tau\,A \expto \T\,\tau\,B)
\]
satisfying time-graded analogues of enriched functor laws
\[
\small
\xymatrix@C=5em@R=4.5em@M=0.5em{
\One \times \T\, \tau\, A \ar[d]_{\snd} \ar[r]^-{\mathsf{curry}(\snd) \times \id} & (A \expto A) \times \T\, \tau\, A \ar[d]^{\etaF_{\!A \expto A,\tau} \times \id}
\\
\T\, \tau\, A & \futuremod{\tau}(A \expto A) \times \T\, \tau\, A \ar[l]^-{\mathsf{uncurry}(\enr_{\!A,A,\tau})}
}
\]
\[
\small
\xymatrix@C=6em@R=4.5em@M=0.5em{
\futuremod{\tau}(B \expto C) \times \big(\futuremod{\tau}(A \expto B) \times T\, \tau\, A \big) \ar[r]^-{\id \times \mathsf{uncurry}(\enr_{\!A,B,\tau})} \ar[d]_{\alpha^{-1}} & \futuremod{\tau}(B \expto C) \times T\, \tau\, B \ar[ddd]^{\mathsf{uncurry}(\enr_{\!B,C,\tau})}
\\
\big(\futuremod{\tau}(B \expto C) \times \futuremod{\tau}(A \expto B)\big) \times T\, \tau\, A \ar[d]_{\m_{B \expto C, A \expto B, \tau} \times \id}
\\
\futuremod{\tau}((B \expto C) \times (A \expto B)) \times T\, \tau\, A \ar[d]_{\futuremod{\tau}(\mathsf{comp}) \times \id}
\\
\futuremod{\tau}(A \expto C) \times T\, \tau\, A \ar[r]_-{\mathsf{uncurry}(\enr_{\!A,C,\tau})} & T\, \tau\, C
}
\]
and interacting well with the unit and multiplication of $\T$
\[
\small
\xymatrix@C=5em@R=4.5em@M=0.5em{
(A \expto B) \times A \ar[r]^-{\varepsilon^{-1}_{A \expto B} \times \etaT_{A}} \ar[d]_{\mathsf{uncurry}(\id)} & \futuremod{0}(A \expto B) \times \T\, 0\, A \ar[d]^{\mathsf{uncurry}(\enr_{\!A,B,0})}
\\
B \ar[r]_-{\etaT_{B}} & \T\, 0\, B
}
\]
\[
\scriptsize
\xymatrix@C=3em@R=4.5em@M=0.5em{
\futuremod{\tau_1}(\futuremod{\tau_2}(A \expto B)) \times \T\, \tau_1\, (\T\, \tau_2\, A) \ar[r]^-{\delta^{-1}_{A \expto B,\tau_1,\tau_2} \times \muT_{A,\tau_1,\tau_2}} \ar[d]_{\futuremod{\tau_1}(\mathsf{curry}(\id)) \times \id } & \futuremod{\tau_1 + \tau_2}(A \expto B) \times \T\, (\tau_1 + \tau_2)\, A \ar[ddd]^{\mathsf{uncurry}(\enr_{\!A,B,\tau_1 + \tau_2})}
\\
\futuremod{\tau_1}\big(\T\, \tau_2\, A \expto (\futuremod{\tau_2}(A \expto B) \times \T\, \tau_2\, A)\big) \times \T\, \tau_1\, (\T\, \tau_2\, A) \ar[d]_{\mathsf{uncurry}(\enr)}
\\
\T\, \tau_1\, (\futuremod{\tau_2}(A \expto B) \times \T\, \tau_2\, A) \ar[d]_{\T\, \tau_1\, (\mathsf{uncurry}(\enr_{\!A,B,\tau_2}))}
\\
\T\, \tau_1\, (\T\, \tau_2\, A) \ar[r]_{\muT_{A,\tau_1,\tau_2}} & \T\, (\tau_1 + \tau_2)\, B
}
\]
where the unit-like $\etaF_{A,\tau}$ is given by the composite
$A \overset{\varepsilon^{-1}_{A}}{-\!\!\!-\!\!\!\to} \futuremod{0} A \overset{\futuremod{0 \le \tau}_A}{-\!\!\!-\!\!\!-\!\!\!\to} \futuremod{\tau}A$, 
and where the composition morphism $\mathsf{comp}_{A,B,C} : (B \expto C) \times (A \expto B) \to A \expto C$ 
is derived from the universal property of exponentials in the standard way~\cite{Kelly:EnrichedCats}.

Given $\futuremod -$-strength, $\futuremod -$-enrichment is given as
\[
\enr_{A,B,\tau} \defeq \futuremod \tau(A \expto B) \overset{\mathsf{curry}\big(\T\, \tau\, (\mathsf{uncurry}(\id)) \,\circ\, \strT_{A \expto B, A, \tau}\big)}{-\!\!\!-\!\!\!-\!\!\!-\!\!\!-\!\!\!-\!\!\!-\!\!\!-\!\!\!-\!\!\!-\!\!\!-\!\!\!-\!\!\!-\!\!\!-\!\!\!-\!\!\!-\!\!\!-\!\!\!-\!\!\!-\!\!\!-\!\!\!-\!\!\!-\!\!\!-\!\!\!-\!\!\!\to} (\T\,\tau\,A \expto \T\,\tau\,B)
\] 
Conversely, given $\futuremod -$-enrichment, $\futuremod -$-strength is defined as
\[
\strT_{A,B,\tau} \defeq \futuremod \tau A \times \T\, B\, \tau \overset{\mathsf{uncurry}(\enr_{\!B,A \times B,\tau}) \,\circ\, \futuremod{\tau}(\mathsf{curry}(\id)) \times \id}{-\!\!\!-\!\!\!-\!\!\!-\!\!\!-\!\!\!-\!\!\!-\!\!\!-\!\!\!-\!\!\!-\!\!\!-\!\!\!-\!\!\!-\!\!\!-\!\!\!-\!\!\!-\!\!\!-\!\!\!-\!\!\!-\!\!\!-\!\!\!-\!\!\!-\!\!\!-\!\!\!-\!\!\!\to} \T\, (A \times B)\, \tau
\]

\newpage

\subsection{Algebraic Operations for Algebraic Effects}

We require natural transformations
\[
\opT_{A,\tau} : \gsem{A_\op} \times \futuremod{\tau_\op}(\gsem{B_\op} \expto \T\, \tau\, A) \to \T\, (\tau_\op + \tau)\, A
\]
for all $\op : \tysigop{A_\op}{B_\op}{\tau_\op} \in \Ops$, and
\[
\delayT_{\!A,\tau'}\, \tau : \futuremod {\tau} (\T\, \tau'\, A) \to \T\, (\tau + \tau')\, A
\]
for all $\tau \in \mathbb{B}$, satisfying algebraicity laws\footnote{For better readability, 
we write the semantics of ground types $\gsem{A_\op}$ and $\gsem{B_\op}$ simply as as $A_\op$ and $B_\op$ 
in the diagrams here and in the rest of the appendices.} with respect to $\muT$ and $\strT$
\[
\small
\xymatrix@C=4em@R=4.5em@M=0.5em{
A_\op \times \futuremod{\tau_\op}(B_\op \expto \T\, \tau\, (\T\, \tau' \, A)) \ar[d]_{\id \times \futuremod{\tau_\op}(B_\op \expto \muT_{A,\tau,\tau'})} \ar[r]^-{\opT_{\T \tau' A,\tau}} & \T\, (\tau_\op + \tau)\, (\T\, \tau' \, A) \ar[dd]^{\muT_{A,\tau_\op + \tau,\tau'}}
\\
A_\op \times \futuremod{\tau_\op}(B_\op \expto \T\, (\tau + \tau')\, A) \ar[d]_{\opT_{A,\tau + \tau'}}
\\
\T\, (\tau_\op + (\tau + \tau'))\, A \ar[r]_-{\equiv}
& \T\, ((\tau_\op + \tau) + \tau')\, A
}
\]
\[
\small
\xymatrix@C=6em@R=4.5em@M=0.5em{
\futuremod{\tau}(\T\, \tau'\, (\T\, \tau''\, A)) \ar[r]^{\delayT_{(\T\, \tau''\, A),\tau'}\, \tau} \ar[d]_{\futuremod{\tau}(\muT_{A,\tau',\tau''})} & \T\, (\tau + \tau')\, (\T\, \tau''\, A) \ar[dd]^{\muT_{A,\tau + \tau', \tau''}}
\\
\futuremod{\tau}(\T\, (\tau' + \tau'') \, A) \ar[d]_{\delayT_{\!A,\tau' + \tau''}\, \tau}
\\
\T\, (\tau + (\tau' + \tau''))\, A \ar[r]_-{\equiv}
& \T\, ((\tau + \tau') + \tau'')\, A
}
\]
\[
\small
\xymatrix@C=4em@R=4em@M=0.5em{
\futuremod{\tau_\op + \tau} A \times (A_\op \times \futuremod{\tau_\op}(B_\op \expto \T\, \tau\, B)) \ar[r]^-{\id \times \opT_{B,\tau}} \ar[d]_{\alpha \,\circ\, (\mathsf{swap} \times \id) \,\circ\, \alpha^{-1}} & \futuremod{\tau_\op + \tau} A \times \T\, (\tau_\op + \tau)\, B \ar[ddddd]^{\strT_{A,B,\tau_\op + \tau}}
\\
A_\op \times (\futuremod{\tau_\op + \tau}A \times \futuremod{\tau_\op}(B_\op \expto \T\, \tau\, B)) \ar[d]_{\id \times (\delta_{A,\tau_\op,\tau} \times \id)}
\\
A_\op \times \futuremod{\tau_\op}(\futuremod{\tau}A) \times \futuremod{\tau_\op}(B_\op \expto \T\, \tau\, B) \ar[d]_{\id \times \m_{\futuremod{\tau}A, B_\op \expto \T \tau B ,\tau_\op}}
\\
A_\op \times \futuremod{\tau_\op}(\futuremod{\tau}A \times (B_\op \expto \T\, \tau\, B)) \ar[d]_{\id \times \futuremod{\tau_\op}(\mathsf{push\text{-}under\text{-}}\expto)}
\\
A_\op \times \futuremod{\tau_\op}(B_\op \expto (\futuremod{\tau}A \times \T\, \tau\, B)) \ar[d]_{\id \times \futuremod{\tau_\op}(B_\op \expto \strT_{A,B,\tau})}
\\
A_\op \times \futuremod{\tau_\op}(B_\op \expto \T\, \tau\, (A \times B)) \ar[r]_-{\opT_{A \times B,\tau}}
& \T\, (\tau_\op + \tau)\, (A \times B)
}
\]
\[
\small
\xymatrix@C=5em@R=4em@M=0.5em{
\futuremod{\tau + \tau'} A \times \futuremod{\tau}(\T\, \tau'\, B) \ar[r]^-{\id \times \delayT_{\!B,\tau'}\, \tau} \ar[d]_{\delta_{A,\tau,\tau'} \times \id} & \futuremod{\tau + \tau'} A \times \T\, (\tau + \tau')\, B \ar[ddd]^{\strT_{A, B, \tau + \tau'}}
\\
\futuremod{\tau}(\futuremod{\tau'}A) \times \futuremod{\tau}(\T\, \tau'\, B) \ar[d]_{\m_{\futuremod{\tau'}A,\T\, \tau'\, B,\tau}}
\\
\futuremod{\tau}(\futuremod{\tau'}A \times \T\, \tau'\, B) \ar[d]_{\futuremod{\tau}(\strT_{A,B,\tau'})}
\\
\futuremod{\tau}(T\, \tau'\, (A \times B)) \ar[r]_-{\delayT_{\!A \times B,\tau'}\, \tau} & \T\, (\tau + \tau')\, (A \times B)
}
\]
where $\mathsf{push\text{-}under\text{-}}\!\!\expto$ is a composite morphism that pushes the given 
argument $\futuremod{\tau}A$ under the exponential---it is defined using the universal property of $\expto$.

\newpage

\subsection{Graded $\T$-Algebras for Effect Handling}

We require morphisms
\[
\begin{array}{l}
\h_{A,\tau,\tau'} : \Pi_{\op \in \Ops} \Pi_{\tau'' \in \mathbb{N}} \big( 
  (\sem{A_\op} \times \futuremod{\tau_\op}(\sem{B_\op} \expto \T\, \tau''\, A)) \expto \T\, (\tau_\op + \tau'')\, A \big)
\\[0.5ex]
\hspace{7cm}
\to \T\, \tau\, (\T\, \tau'\, A) \expto \T\, (\tau + \tau')\, A
\end{array}
\] 
satisfying laws stating that $\h_A$ returns a graded $\T$-algebra for $\opT$ and $\delayT$
\[
\small
\xymatrix@C=4em@R=4.5em@M=0.5em{
\mathcal{H} \times \T\, \tau\, A \ar[r]^-{\id \times \etaT_{\T\, \tau\, A}} \ar[d]_{\snd} & \mathcal{H} \times \T\, 0\, (\T\, \tau\, A) \ar[d]^{\mathsf{uncurry}(\h_{A,0,\tau})}
\\
\T\, \tau\, A \ar[r]_-{\equiv} & \T\, (0 + \tau)\, A
}
\]
\[
\small
\xymatrix@C=4.8em@R=4.5em@M=0.5em{
\mathcal{H} \times \futuremod{\tau}(\T\, \tau'\, (\T\, \tau''\, A)) \ar[r]^-{\id \times \delayT_{\T \tau'' A,\tau'}\, \tau} \ar[d]_{\etaA_{\mathcal{H},\tau} \times \id} & \mathcal{H} \times \T\, (\tau + \tau')\, (\T\, \tau''\, A) \ar[ddddd]^{\mathsf{uncurry}(\h_{\!A,\tau + \tau'\!,\tau''})}
\\
\futuremod{\tau}(\pastmod{\tau} \mathcal{H}) \times \futuremod{\tau}(\T\, \tau'\, (\T\, \tau''\, A)) \ar[d]_{\m_{\pastmod{\tau} \mathcal{H}, \T\, \tau'\, (\T\, \tau''\, A), \tau}}
\\
\futuremod{\tau}(\pastmod{\tau} \mathcal{H} \times \T\, \tau'\, (\T\, \tau''\, A)) \ar[d]_{\futuremod{\tau}(\epsP_{\mathcal{H},\tau} \times \id)}
\\
\futuremod{\tau}(\mathcal{H} \times \T\, \tau'\, (\T\, \tau''\, A)) \ar[d]_{\futuremod{\tau}(\mathsf{uncurry}(\h_{A,\tau',\tau''}))}
\\
\futuremod{\tau}(\T\, (\tau' + \tau'')\, A) \ar[d]_{\delayT_{\!A,\tau' + \tau''}\, \tau}
\\
\T\, (\tau + (\tau' + \tau''))\, A \ar[r]_{\equiv}
& \T\, ((\tau + \tau') + \tau'')\, A
}
\]
\[
\hspace{-0.2cm}
\scriptsize
\xymatrix@C=1em@R=4em@M=0.5em{
\mathcal{H} \times (A_\op \times \futuremod{\tau_\op}(B_\op \expto \T\, \tau\, (\T\, \tau'\, A))) \ar[r]^-{\id \times \opT_{\T \tau' A,\tau}} \ar[d]_{\langle \fst , \id \rangle} & \mathcal{H} \times \T\, (\tau_\op + \tau)\, (\T\, \tau'\, A) \ar[ddddddddd]^{\mathsf{uncurry}(\h_{A,\tau_\op + \tau,\tau'})}
\\
\mathcal{H} \times (\mathcal{H} \times (A_\op \times \futuremod{\tau_\op}(B_\op \expto \T\, \tau\, (\T\, \tau'\, A)))) \ar[d]_{(\mathsf{proj}_{\tau + \tau'} \,\circ\, \mathsf{proj}_{\op}) \times \id}
\\
\mathcal{H}_{\op,\tau + \tau'} \times (\mathcal{H} \times (A_\op \times \futuremod{\tau_\op}(B_\op \expto \T\, \tau\, (\T\, \tau'\, A)))) \ar[d]_{\id \times (\alpha \,\circ\, (\mathsf{swap} \times \id) \,\circ\, \alpha^{-1})}
\\
\mathcal{H}_{\op,\tau + \tau'} \times (A_\op \times (\mathcal{H} \times \futuremod{\tau_\op}(B_\op \expto \T\, \tau\, (\T\, \tau'\, A)))) \ar[d]_{\id \times (\id \times (\etaA_{\mathcal{H},\tau_\op} \times \id))}
\\
\mathcal{H}_{\op,\tau + \tau'} \times (A_\op \times (\futuremod{\tau_\op}(\pastmod{\tau_\op}\mathcal{H}) \times \futuremod{\tau_\op}(B_\op \expto \T\, \tau\, (\T\, \tau'\, A))))) \ar[d]_{\id \times (\id \times \m_{\pastmod{\tau_\op}\mathcal{H}, B_\op \expto \T \tau (\T \tau' A), \tau_\op})}
\\
\mathcal{H}_{\op,\tau + \tau'} \times (A_\op \times \futuremod{\tau_\op}(\pastmod{\tau_\op}\mathcal{H} \times (B_\op \expto \T\, \tau\, (\T\, \tau'\, A)))) \ar[d]_{\id \times (\id \times \futuremod{\tau_\op}(\epsP_{\mathcal{H},\tau} \times \id))}
\\
\mathcal{H}_{\op,\tau + \tau'} \times (A_\op \times \futuremod{\tau_\op}(\mathcal{H} \times (B_\op \expto \T\, \tau\, (\T\, \tau'\, A)))) \ar[d]_{\id \times (\id \times \futuremod{\tau_\op}(\mathsf{push\text{-}under\text{-}}\expto))}
\\
\mathcal{H}_{\op,\tau + \tau'} \times (A_\op \times \futuremod{\tau_\op}(B_\op \expto (\mathcal{H} \times \T\, \tau\, (\T\, \tau'\, A)))) \ar[d]_{\id \times (\id \times \futuremod{\tau_\op}(B_\op \expto \mathsf{uncurry}(\h_{A,\tau,\tau'})))}
\\
\mathcal{H}_{\op,\tau + \tau'} \times (A_\op \times \futuremod{\tau_\op}(B_\op \expto \T\, (\tau + \tau')\, A)) \ar[d]_{\mathsf{uncurry}\, \id}
\\
\T\, (\tau_\op + (\tau + \tau'))\, A \ar[r]_{\equiv}
& \T\, ((\tau_\op + \tau) + \tau')\, A
}
\vspace{0.5cm}
\]
where we write $\mathcal{H}$ for the domain of the morphisms $\h_{A,\tau,\tau'}$, 
$\mathcal{H}_{\op,\tau + \tau'}$ for the operation case
$(\sem{A_\op} \times \futuremod{\tau_\op}(\sem{B_\op} \expto \T\, (\tau + \tau')\, A)) \expto \T\, (\tau_\op + (\tau + \tau'))\, A$, 
and where the counit-like $\epsP_{A,\tau}$ is given by the composite 
$\pastmod{\tau} A \overset{\pastmod{0 \le \tau}_A}{-\!\!\!-\!\!\!-\!\!\!-\!\!\!\to} 
  \pastmod{0} A \overset{\eta^{-1}_A}{-\!\!\!\to} A$.

\newpage

\subsection{Interpretation of Values and Computations}

Well-typed values are interpreted as follows
\[
\small
%\hspace{-2cm}
\begin{array}{l}
\sem{\Gamma,x \of X,\Gamma' \types x : X} 
~\defeq~ 
\sem{\Gamma,x \of X,\Gamma'}\One
\overset{\iota}{\to}
\sem{\Gamma'}\big(\sem{\Gamma}\One \times \sem{X}\big)
\\[0.5ex]
\hspace{3.75cm}
\overset{\e}{-\!\!\!\!\!\!\longrightarrow}
\pastmod{\ctxtime {\Gamma'}}\big(\sem{\Gamma}\One \times \sem{X}\big)
\overset{\epsP}{-\!\!\!\!\!\!\longrightarrow}
\sem{\Gamma}\One \times \sem{X}
\overset{\snd}{-\!\!\!\!\!\!\longrightarrow}
\sem{X}
\\[3ex]
\sem{\Gamma \types \tmconst{f}(V_1, \ldots , V_n) : B}
~\defeq~
\sem{\Gamma}\One 
\overset{\langle \sem{V_1} , \ldots , \sem{V_n} \rangle}{-\!\!\!-\!\!\!-\!\!\!-\!\!\!-\!\!\!-\!\!\!-\!\!\!-\!\!\!-\!\!\!\longrightarrow} 
\sem{A_1} \times \ldots \times \sem{A_n} 
\overset{\sem{\tmconst{f}}}{-\!\!\!\!\!\!\longrightarrow}
\sem{B}
\\[3ex]
\sem{\Gamma \types \tmpair{V}{W} : \typrod{X}{Y}}
~\defeq~
\sem{\Gamma}\One 
\overset{\langle \sem{V} , \sem{W} \rangle}{-\!\!\!-\!\!\!-\!\!\!-\!\!\!-\!\!\!-\!\!\!\longrightarrow}
\sem{X} \times \sem{Y}
\\[3ex]
\sem{\Gamma \types \tmunit : \tyunit}
~\defeq~
\sem{\Gamma}\One
\overset{!}{-\!\!\!\!\!\!\longrightarrow}
\One
\\[3ex]
\sem{\Gamma \types \tmfun{x : X}{M} : \tyfun{X}{\tycomp{Y}{\tau}}}
~\defeq~
\sem{\Gamma}\One
\overset{\mathsf{curry}(\sem{M})}{-\!\!\!-\!\!\!-\!\!\!-\!\!\!-\!\!\!\longrightarrow}
\sem{X} \expto \T\, \tau\, \sem{Y}
\\[3ex]
\sem{\Gamma \types \tmbox[\tau]{V} : \tybox{\tau}{X}} 
~\defeq~
\sem{\Gamma}\One
\overset{\etaA}{-\!\!\!\!\!\!\longrightarrow}
\futuremod{\tau} \big(\pastmod{\tau}(\sem{\Gamma}\One)\big)
\overset{\futuremod{\tau}(\sem{V})}{-\!\!\!-\!\!\!-\!\!\!\longrightarrow}
\futuremod{\tau} \sem{X}
\end{array}
\vspace{0.5cm}
\]

\noindent
Well-typed computations are interpreted as follows 
\[
\small
\begin{array}{l}
\sem{\Gamma \types \tmreturn{V} : \tycomp{X}{0}}
~\defeq~
\sem{\Gamma}\One
\overset{\sem{V}}{-\!\!\!\!\!\!\longrightarrow}
\sem{X}
\overset{\etaT}{-\!\!\!\!\!\!\longrightarrow}
\T\, 0\, \sem{X}
\\[4ex]
\sem{\Gamma \types \tmlet{x}{M}{N} : \tycomp{Y}{\tau + \tau'}}
~\defeq~
\sem{\Gamma} \One
\overset{\langle \etaA , \sem{M} \rangle}{-\!\!\!-\!\!\!-\!\!\!-\!\!\!-\!\!\!\longrightarrow}
\futuremod{\tau}\big(\pastmod{\tau}(\sem{\Gamma} \One)\big) \times \T\,\tau\,\sem{X}
\\[1ex]
\hspace{2cm}
\overset{\strT}{-\!\!\!\!\!\!\longrightarrow}
\T\, \tau\, \big(\pastmod{\tau}(\sem{\Gamma} \One) \times \sem{X}\big)
\overset{\T\,(\sem{N})}{-\!\!\!-\!\!\!-\!\!\!-\!\!\!\longrightarrow}
\T\, \tau\, (\T\, \tau'\, \sem{Y})
\overset{\muT}{-\!\!\!\!\!\!\longrightarrow}
\T\, (\tau + \tau')\, \sem{Y}
\\[4ex]
\sem{\Gamma \types \tmapp{V}{W} : \tycomp{Y}{\tau}}
~\defeq~
\sem{\Gamma} \One 
\overset{\langle \sem{V} , \sem{W} \rangle}{-\!\!\!-\!\!\!-\!\!\!-\!\!\!-\!\!\!-\!\!\!\longrightarrow}
(\sem{X} \expto \T\, \tau\, \sem{Y}) \times \sem{X}
\overset{\mathsf{uncurry}(\id)}{-\!\!\!-\!\!\!-\!\!\!-\!\!\!\longrightarrow}
\T\, \tau\, \sem{Y}
\\[4ex]
\sem{\Gamma \types \tmmatch{V}{\tmpair{x}{y} \mapsto M} : \tycomp{Z}{\tau}}
~\defeq~
\sem{\Gamma} \One
\overset{\langle \id , \sem{V} \rangle}{-\!\!\!-\!\!\!-\!\!\!\longrightarrow}
\sem{\Gamma} \One \times (\sem{X} \times \sem{Y})
\\[1ex]
\hfill
\overset{\alpha^{-1}}{-\!\!\!\!\!\!\longrightarrow}
(\sem{\Gamma} \One \times \sem{X}) \times \sem{Y}
\overset{\sem{M}}{-\!\!\!\!\!\!\longrightarrow}
\T\, \tau\, \sem{Z}
\\[4ex]
\sem{\Gamma \types \tmop{op}{V}{\tmcont x M} : \tycomp{X}{\tau_\op + \tau}}
~\defeq~
\sem{\Gamma} \One
\overset{\langle \sem{V} , \etaA \rangle}{-\!\!\!-\!\!\!-\!\!\!-\!\!\!-\!\!\!\longrightarrow}
\gsem{A_\op} \times \futuremod{\tau_\op}\big(\pastmod{\tau_\op}(\sem{\Gamma} \One)\big)
\\[1ex]
\hfill
\overset{\!\!\id \times \futuremod{\tau_\op}(\mathsf{curry}(\sem{M}))}{-\!\!\!-\!\!\!-\!\!\!-\!\!\!-\!\!\!-\!\!\!-\!\!\!-\!\!\!-\!\!\!-\!\!\!-\!\!\!-\!\!\!\longrightarrow\,\,}
\gsem{A_\op} \times \futuremod{\tau_\op}\big(\gsem{B_\op} \expto \T\, \tau\, \sem{X}\big)
\overset{\opT}{-\!\!\!\!\!\!\longrightarrow}
\T\, (\tau_\op + \tau)\, \sem{X}
\\[4ex]
\sem{\Gamma \types \tmdelay{\tau}{M} : \tycomp{X}{\tau + \tau'}}
~\defeq~
\sem{\Gamma} \One
\overset{\etaA}{-\!\!\!\!\!\!\longrightarrow}
\futuremod{\tau}\big(\pastmod{\tau}(\sem{\Gamma} \One)\big)
\\[1ex]
\hfill
\overset{\futuremod{\tau}(\sem{M})}{-\!\!\!-\!\!\!-\!\!\!-\!\!\!-\!\!\!\!\!\!\longrightarrow}
\futuremod{\tau}(\T\, \tau'\, \sem{X})
\overset{\delayT\, \tau}{-\!\!\!-\!\!\!-\!\!\!\longrightarrow}
\T\, (\tau + \tau')\, \sem{X}
\end{array}
\]

\[
\begin{array}{l}
\sem{\Gamma \types \tmhandle{M}{H}{x}{N} : \tycomp{Y}{\tau + \tau'}}
~\defeq~
\\[1ex]
\hspace{0.3cm}
\sem{\Gamma} \One
\overset{\langle \id , \langle \etaA , \sem{M} \rangle \rangle}{-\!\!\!-\!\!\!-\!\!\!-\!\!\!-\!\!\!-\!\!\!-\!\!\!-\!\!\!\longrightarrow}
\sem{\Gamma} \One \times \Big(\futuremod{\tau}\big(\pastmod{\tau}(\sem{\Gamma} \One)\big) \times \T\, \tau\, \sem{X}\Big)
\\[1ex]
\hspace{0.7cm}
\overset{\!\!\id \times \strT}{-\!\!\!-\!\!\!\longrightarrow\,\,}
\sem{\Gamma} \One \times \T\, \tau\, \big(\pastmod{\tau}(\sem{\Gamma} \One) \times \sem{X} \big)
\overset{\id \times \T\, \tau\, (\sem{N})}{-\!\!\!-\!\!\!-\!\!\!-\!\!\!-\!\!\!-\!\!\!-\!\!\!\longrightarrow}
\sem{\Gamma} \One \times \T\, \tau\, \big(\T\, \tau'\, \sem{Y}\big)
\\[1ex]
\hspace{3.4cm}
\overset{\sem{H} \times \id}{-\!\!\!-\!\!\!-\!\!\!\longrightarrow}
\mathcal{H} \times \T\, \tau\, \big(\T\, \tau'\, \sem{Y}\big)
\overset{\mathsf{uncurry}(\h_{\sem{Y},\tau,\tau'})}{-\!\!\!-\!\!\!-\!\!\!-\!\!\!-\!\!\!-\!\!\!-\!\!\!-\!\!\!-\!\!\!-\!\!\!\longrightarrow}
\T\, (\tau + \tau')\, \sem{Y}
\\[4ex]
\sem{\Gamma \types \tmunbox[\tau] V x N : \tycomp{Y}{\tau'}}
~\defeq~
\sem{\Gamma}\One 
\overset{\langle \id , \etaPRA \rangle}{-\!\!\!-\!\!\!-\!\!\!-\!\!\!\longrightarrow}
\sem{\Gamma}\One \times \pastmod{\tau}\big(\sem{\Gamma \ctxminus \tau} \One\big)
\\[1ex]
\hfill
\overset{\!\!\id \times \pastmod{\tau}(\sem{V})}{-\!\!\!-\!\!\!-\!\!\!-\!\!\!-\!\!\!-\!\!\!\longrightarrow\,\,}
\sem{\Gamma}\One \times \pastmod{\tau}\big(\futuremod{\tau} \sem{X}\big)
\overset{\id \times \epsA}{-\!\!\!-\!\!\!\longrightarrow}
\sem{\Gamma}\One \times \sem{X}
\overset{\sem{N}}{-\!\!\!\!\!\!\longrightarrow}
\T\,\tau'\, \sem{Y}
\end{array}
\]
where in the $\tmhandle{M}{H}{x}{N}$ case we write $\mathcal{H}$ for 
\[
\Pi_{\op \in \Ops} \Pi_{\tau'' \in \mathbb{N}} \big( 
  (\gsem{A_\op} \times \futuremod{\tau_\op}(\gsem{B_\op} \expto \T\, \tau''\, \sem{Y})) \expto \T\, (\tau_\op + \tau'')\, \sem{Y} \big)
\]
and $\sem{H}$ for
\[
\sem{\Gamma}\One
\overset{\langle \langle \id \rangle_{\tau'' \in \mathbb{N}} \rangle_{\op \in \Ops}}{-\!\!\!-\!\!\!-\!\!\!-\!\!\!-\!\!\!-\!\!\!-\!\!\!-\!\!\!\longrightarrow}
\Pi_{\op \in \Ops} \Pi_{\tau'' \in \mathbb{N}} \Big(\sem{\Gamma}\One\Big)
\overset{\Pi_{\op \in \Ops} \Pi_{\tau'' \in \mathbb{N}} \big( \mathsf{curry}(\sem{M_\op\, \tau''} \,\circ\, \alpha^{-1}) \big)}{-\!\!\!-\!\!\!-\!\!\!-\!\!\!-\!\!\!-\!\!\!-\!\!\!-\!\!\!-\!\!\!-\!\!\!-\!\!\!-\!\!\!-\!\!\!-\!\!\!-\!\!\!-\!\!\!-\!\!\!-\!\!\!-\!\!\!-\!\!\!-\!\!\!-\!\!\!-\!\!\!\longrightarrow}
\mathcal{H}
\]
and where we recall from \cref{sect:core-calculus} that we write $H$ for ${\tmopclause{x}{k}{M_{\op}}}_{\op \in \Ops} $.

\end{document}